\documentclass[a4paper,12pt]{amsart}

\usepackage{amsthm}
\usepackage{amsmath}
\usepackage{geometry}
\usepackage{amsfonts}
\usepackage{graphicx}
\usepackage{array}
\usepackage{amssymb}
\usepackage{mathtools}

\newenvironment{alphenumerate}
				{ \begin{enumerate}
				}
				{\end{enumerate}
				}

\newcommand{\scalar}[2]{\langle#1\,,#2\rangle}
\newcommand{\pair}[2]{(#1\,,#2)}
\newcommand{\norm}[1]{\|#1\|}
\newcommand{\normm}[1]{|||#1|||}
\renewcommand{\H}{\mathcal{H}}
\renewcommand{\L}{\mathcal{L}}
\newcommand{\D}{\mathcal{D}}
\renewcommand{\d}{\mathrm{d}}
\newcommand{\pO}{{\partial \Omega}}

\newcommand{\C}{\mathcal{C}}
\newcommand{\1}{\mathbb{I}}
\newcommand{\R}{\mathbb{R}}
\newcommand{\bx}{\mathbf{x}}

\newcommand{\bk}{\mathbf{k}}

\newcommand{\dom}{\mathrm{Dom}}

\newcommand{\map}[5][\phantom]{
	\begin{array}{ccc}
		\mathllap{#1{\,:\,}}{#2}&\to&{#3}\\
			{#4}&\rightsquigarrow&{#5}
	\end{array}	
	}

\newtheorem{theorem}{Theorem}
\newtheorem{problem}[theorem]{Problem}

\newtheorem{definition}[theorem]{Definition}
\newtheorem{lemma}[theorem]{Lemma}

\newcounter{myexample}[section]
\newenvironment{example}[1][]{
			\refstepcounter{myexample}
			\vspace{0.4cm}
			\noindent 
			{\bf Example \arabic{section}.\arabic{myexample}\normalfont{#1}{\bf.}}
			\noindent
			}
			{
			\hfill$\blacksquare$\newline
			}

\numberwithin{equation}{section}
\numberwithin{theorem}{section}
\numberwithin{myexample}{section}

\begin{document}

\title[S.A. extensions of symmetric op. and applications to Quantum Physics]{On the theory of self-adjoint extensions of symmetric operators and its applications to Quantum Physics}

\author{A. Ibort}
\address{ICMAT and Depto. de Matem\'aticas, Univ. Carlos III de Madrid, Avda. de la
Universidad 30, 28911 Legan\'es, Madrid, Spain.}
\email{albertoi@math.uc3m.es}
\author{J.M. P\'{e}rez-Pardo}
\address{INFN - Sezione di Napoli, Complesso Universitario di Monte Sant'Angelo, via Cintia, Edificio 6, 80126 Naples, Italy}
\email{juanma@na.infn.it}

\begin{abstract}
This is a series of 5 lectures around the common subject of the construction of self-adjoint extensions of symmetric operators and its applications to Quantum Physics. We will try to offer a brief account of some recent ideas in the theory of self-adjoint extensions of symmetric operators on Hilbert spaces and their applications to a few specific problems in Quantum Mechanics.   
\end{abstract}

\maketitle
\tableofcontents


\section{Introduction}\label{sec:introduction}

In this series we do not pretend to offer a review of the basic theory of self-adjoint extensions of symmetric operators which is already well-known and has been described extensively in various books (see for instance \cite{reed75}, \cite{akhiezer61b}, \cite{We80}).  
Instead, what we try to offer the reader in these notes is a selection of problems inspired by Quantum Mechanics  where the study of self-adjoint extensions of symmetric operators constitutes a basic ingredient.   

The reader may find in the set of lectures \cite{Ib12} a recent discussion on the theory of self-adjoint extensions of Laplace-Beltrami and Dirac operators in manifolds with boundary, as well as a family of examples and applications.  In a sense, the present lectures can be considered as a follow up.  

Thus, in the current series, Lecture I will be devoted to analyse the problem of determining self-adjoint extensions of bipartite systems whose components are defined by symmetric operators.  An application of such ideas to control entangled states will be discussed.  Lecture II will cover the recent approach developed in \cite{ibortlledo13} to the theory of self-adjoint extensions of the Laplace-Beltrami operator from the point of view of their corresponding quadratic forms.   New results on the theory as well as the introduction of the notion of admissible unitary operators at the boundary will be discussed.   Lecture III will be devoted to discuss the fundamentals of the theory of self-adjoint extensions of Dirac-like operators in manifolds with boundary and Lecture IV will cover the study of the construction of self-adjoint extensions for operators defining non-semibounded quadratic forms like the Dirac operators considered in the previous lecture.
Finally, Lecture V will cover the situation where symmetries are present. This discussion is closely related to the results presented in \cite{ibortlledo14}.  Self-adjoint extensions with symmetry will be described both in von Neumann's picture and using the theory of quadratic forms.  Explicit examples will be provided for the Laplace-Beltrami operator.


\newpage

\section{Lecture I. Self-adjoint extensions of bipartite systems}\label{sec:bipartite}

The first lecture is devoted to discuss the problem of self-adjoint extensions of bipartite systems where one or both of the individual systems are defined by symmetric but not self-adjoint operators and the complete description of the bipartite system requires a self-adjoint extension of the overall system to be fixed.    This situation will arise, for instance, whenever we have a quantum system which is coupled or controlled by another quantum system like a particle in a box or in a trap, where the boundary conditions determine the self-adjoint extension of the system.  However, and this is the problem we will analyse here, are these all the possibilities for such a system or, as it often happens in quantum systems, will the tensor product unwind other possibilities and there are other self-adjoint extensions of the bipartite system that will go beyond the individual ones?
The answer, as it will be shown later, is positive and we discuss some simple applications. In particular it is shown that an appropriate choice of self-adjoint extensions of the bipartite system, would allow to generate entangled states.

Before embarking in the formal definitions we would like to pose more precisely the main problem that will be analysed in this lecture showing the vast field of possibilities that arise when considering the set of self-adjoint extensions.

\begin{problem}\label{prob:bipartite}
Consider two quantum systems $A$ and $B$, that we shall call auxiliary and bulk system respectively. System $A$ is defined on a Hilbert space $\H_A$ and equivalently system $B$ is defined on $\H_B$. Now consider that the dynamics in each of these systems is not completely determined in the sense of \cite{reed75}, i.e., the dynamics is given in terms of just densely defined symmetric Hamiltonian operators $H_A$ and $H_B$ and a self-adjoint extension needs to be specified in order to define individual proper dynamics. What are the self-adjoint extensions of the composite system $A\otimes B$\,?
\end{problem}

We will provide a partial answer to this problem and a conjecture on the final solution. Since the aim of these lectures is to introduce the topics to a wide audience we will skip the most technical details. Most of them can be found in \cite{Ibort12}.

The situation exhibited in Problem \ref{prob:bipartite} is not the standard one. Usually one considers the same situation but with both Hamiltonian operators being already self-adjoint. In such a case the Hilbert space of the bipartite system becomes the tensor product of the Hilbert spaces of the parties, 
$$\H=\H_A\otimes\H_B$$ and the dynamics is described by the self-adjoint operator 
\begin{equation}\label{eq:sepHamiltonian}
H_{AB}=H_A\otimes\1+\1\otimes H_B\;.
\end{equation}
In such a case the unitary evolution of the system factorizes in terms of the unitary evolution in each of the subsystems, this is $$U_t=\exp it H_{AB}=\exp it H_A\otimes \exp it H_B=U_t^{A}\otimes U_t^{B}\;.$$
This latter kind of dynamics is called separable, i.e., the evolution of the subsystems is independent. 

If one considers now the situation exposed in Problem~\ref{prob:bipartite}, one needs first to select a self-adjoint extension for the symmetric Hamiltonian operator $H_{AB}$ in order to determine the dynamics of the bipartite system. In contrast with what one could expect, there are self-adjoint extensions leading to non-separable dynamics even if the symmetric operator $H_{AB}$ is of the form \eqref{eq:sepHamiltonian}. We analyse this with more detail in the next sections.


\subsection{A simple example: The half-line--spin 1/2 bipartite system}
As our first example we consider the situation where one of the parties, the auxiliary system, is described by a symmetric but not self-adjoint operator. The bulk system, system $B$, is going to be a finite dimensional system, and therefore automatically self-adjoint. The auxiliary system and the bulk system are given as follows:
\begin{itemize}
\item[$A$:] {$\H_A=\L^2\left([0,\infty)\right)\;,\quad H_A=-\frac{\d^2}{\d x^2}\;,\quad \dom(H_A)=\C_c^{\infty}(0,\infty)$\,.}
\item[$B$:] {$\H_B=\mathbb{C}^2\;,\quad H_B$ is a Hermitean matrix with eigenvalues $\lambda_1>\lambda_2$\,.}
\end{itemize}
Notice that $H_A$ is symmetric but not self-adjoint on its domain. Our aim is now to compute all the self-adjoint extensions of $A\times B$ and determine which ones define separable dynamics.

Before discussing it we may wonder if we can say something general in this situation. Actually system $B$ is already self-adjoint. If we determine the set of self-adjoint extensions of system $A$, that we denote by $\mathcal{M}_A$\,, shouldn't the set of self-adjoint extensions of the composite system $\mathcal{M}_{AB}$ be such that $\mathcal{M}_{AB}=\mathcal{M}_A$\,?
It is well known that in this case we have that $\mathcal{M}_A=\mathcal{U}(1)$\,, see for instance \cite{asorey05,ibortlledo13,Ibort13,kochubei75,bruning08}. However, as we will see by means of a computation later, the set of self-adjoint extensions of the composite system is much bigger.

In order to address the problem we can use von Neumann's abstract characterization of the sets of self-adjoint extensions of symmetric operators. The results can be summarized in the following theorem.

\begin{theorem}[von Neumann]\label{thm:vonneumann}
	The set of self-adjoint extensions of a densely defined, symmetric operator $T$ on a complex, separable Hilbert space, $\mathcal{M}(T)$, is in one to one correspondence with the set of unitary operators $\mathcal{U}(\mathcal{N}_+,\mathcal{N_-})$\,, where $$\mathcal{N}_{\pm}=\operatorname{ker}(T^{\dagger}\mp i\1)=\operatorname{ran}(T\pm i\1)^{\bot}$$ are the deficiency spaces of $T$\,.
	
	Moreover, if we denote by $T_K$ the self-adjoint extension determined by the unitary operator $K:\mathcal{N}_+\to\mathcal{N_-}$ we have that $$\dom(T_K)=\dom(T)\oplus(\1+K)\mathcal{N}_+$$ $$T_K\Phi=T\Phi_0+i(\1-K)\xi_+\;,\quad\Phi=\Phi_0+(\1+K)\xi_+\in\dom(T_K)\;.$$
	The self-adjoint extension $T_K$ will be called in what follows the von Neumann extension of $T$ defined by $K$.
\end{theorem}
For more details about this theorem and its proof we refer to \cite[Chapter 2]{akhiezer61b}, \cite[Chapter X]{reed75}.\\

Now we can look at our bipartite system and we get the following theorem.

\begin{theorem}\label{thm:compositeAB}
	Let $A$ and $B$ be two subsystems of the composite system $AB$\,. Let $H_A$, $\H_A$ denote the Hamiltonian operator and the Hilbert space of system $A$ and let $H_B$, $\H_B$ denote the Hamiltonian operator and the Hilbert space of system $B$. Consider that $H_A$ is a symmetric Hamiltonian operator, $H_B$ is a self-adjoint operator and consider that the dynamics of the composite system is given by the symmetric operator $$H_{AB}=H_A\otimes\1+\1\otimes H_B\;.$$ Let $\mathcal{N}_{A\pm}$ be the deficiency spaces of system $A$. Then, the deficiency spaces of the bipartite system $\mathcal{N}_{AB\pm}$ satisfy $$\mathcal{N}_{AB\pm}\simeq\mathcal{N}_{A\pm}\otimes\H_B\;.$$
\end{theorem}

Before going to the proof let us consider first the consequences of this result. Because of von Neumann's theorem, the set of self-adjoint extensions of the bipartite system satisfies $$\mathcal{M}_{AB\pm}=\mathcal{U}(\mathcal{N}_{AB+},\mathcal{N}_{AB-})=\mathcal{U}(\mathcal{N}_{A+}\otimes\H_B,\mathcal{N}_{A-}\otimes\H_B)\supsetneq\mathcal{U}(\mathcal{N}_{A+},\mathcal{N}_{A-})\otimes \mathcal{U}(\H_B)\;.$$
The group at the right hand side is a proper subgroup and will be relevant later on. The self-adjoint extensions belonging to this group will be called \emph{decomposable} extensions.

\begin{proof}
We will assume for simplicity that the spectrum of $H_B$ is discrete and non-degenerate. Let $\{\rho_n\}$ denote the complete orthonormal basis given by the eigenvectors of $H_B$, $$H_B\rho_n=\lambda_n\rho_n\;.$$ Any $\Phi\in\H=\H_A\hat{\otimes}\H_B$ admits a unique decomposition in terms of the former orthonormal base 
$$\Phi=\sum_{k}\Phi_k\otimes\rho_k\;.$$
We use the symbol $\hat{\otimes}$ to denote that the closure with respect to the natural topology on the tensor product is taken.

Let $\Phi\in\mathcal{N}_{AB+}$. Then we have that $$H_{AB}^\dagger\Phi=i\Phi\;$$ $$\sum\left(H^\dagger_A\Phi_k+\lambda_k\right)\otimes\rho_k=\sum i\Phi_k\otimes \rho_k$$ $$H^\dagger_A\Phi_k=(i-\lambda_k)\Phi_k\;.$$
Let $z\in\mathbb{C}$. If we denote $\mathcal{N}_{Az}=\ker(H^\dagger_A-\bar{z}\1)$ the equation above shows that $\Phi_k\in\mathcal{N}_{A(\lambda_k+i)}$\,. The dimension of the deficiency spaces is constant along the upper and lower complex half planes, cf. \cite{reed75}. Therefore, there exists for each $k$ an isomorphism $$\alpha_k:\mathcal{N}_{A(\lambda_k+i)}\to \mathcal{N}_{A+}\;,$$ and we can consider the isomorphism
$$
\map[\alpha]{\mathcal{N}_{AB+}}{\mathcal{N}_{A+}\otimes \H_B}{\Phi}{\sum \alpha_k(\Phi_k)\otimes\rho_k \mathrlap{\in \mathcal{N}_{A+}\otimes \H_B\;.}}
$$
Similarly for $\mathcal{N}_{A-}$\,.
\end{proof}

To fix the ideas we can go back to our example and compute the different deficiency spaces explicitly. First we need to compute the deficiency spaces for subsystem $A$\,,
$$\mathcal{N}_{A\pm}=\ker\left(-\frac{\d^2}{\d x^2}\mp i\right)\;.$$
Hence we need to find solutions of the following equation that lie in $\L^2([0,\infty))$\,.
\begin{equation}
	-\frac{\d^2}{\d x^2}\Phi=\pm i \Phi\;.
\end{equation}
The general solution is $\Phi=C_1 \exp \sqrt{\pm i} x + C_2 \exp -\sqrt{\pm i} x=C_1 \exp \frac{1}{\sqrt{2}}(1\pm i) x + C_2 \exp -\frac{1}{\sqrt{2}}(1\pm i) x$\,. Since the solutions must be in $\L^2([0,\infty))$ the coefficient $C_1=0$\,. This leads us to 
$$
	\mathcal{N}_{A+}=\mathrm{span}\left\{ \exp \left(-\frac{1}{\sqrt{2}}(1+ i) x\right)\right\}\;,
$$
$$
	\mathcal{N}_{A-}=\mathrm{span} \left\{ \exp \left(-\frac{1}{\sqrt{2}}(1- i) x\right)\right\}\;,
$$
$$
	\dim \mathcal{N}_{A+} =\dim \mathcal{N}_{A-} = 1\;.
$$
Therefore we have that $\mathcal{N}_{A\pm}\simeq \mathbb{C}$\,. According to the previous theorem and the fact that in this case $\H_B\simeq \mathbb{C}^2$ we have that $$\mathcal{N}_{AB\pm}\simeq\mathbb{C}\otimes\mathbb{C}^2\simeq \mathbb{C}^2\;.$$ Hence the set of self-adjoint extensions of the composite system $\mathcal{M}_{AB}$ strictly contains the set of self-adjoint extensions of subsystem $A$\,.

$$
	\mathcal{M}_A=\mathcal{U}(\mathcal{N}_{A+},\mathcal{N}_{A-})\simeq \mathcal{U}(1)\;.
$$
$$
	\mathcal{M}_{AB}=\mathcal{U}(\mathcal{N}_{AB+},\mathcal{N}_{AB-})\simeq \mathcal{U}(2)\;.
$$

The previous discussion allows us to pose the following conjecture. In terms of the boundary conditions and for the case of the Laplace-Beltrami operator, which is a straightforward generalization of the one-dimensional case treated here, it is easy to parametrize the space of self-adjoint extensions in terms of the set of unitary operators on the Hilbert spaces of boundary data. Namely, in the case were 
$$
\begin{array}{cc}
\H_A=\L^2(\Omega_A) & \H_B=\L^2(\Omega_B)\\
H_A=-\Delta_{A} & H_B=-\Delta_B
\end{array}\,
$$
the self-adjoint extensions will be determined by the unitary operators acting on the induced Hilbert space at the boundary:
$$
\partial(\Omega_A\times\Omega_B)=\pO_A\times \Omega_B \cup \Omega_A\times \pO_B\;.
$$
The space of boundary data then becomes
$$
\L^2(\partial(\Omega_A\times\Omega_B))=\left(\L^2(\pO_A)\otimes \L^2(\Omega_B)\right) \oplus \left(\L^2(\Omega_A)\otimes \L^2(\pO_B)\right)\;.
$$
In this situation one can identify $\L^2(\partial(\Omega))\simeq \mathcal{N}_{\pm}$ and $\L^2(\Omega)\simeq \H$\,, and hence the above equation can be interpreted in the following way
$$\mathcal{N}_{AB\pm}\stackrel{?}{=}\left(\mathcal{N}_{A\pm}\otimes\H_B\right)\oplus  \left(\H_A\otimes\mathcal{N}_{\pm}\right)\;.$$
This would be the generalization of Theorem \ref{thm:compositeAB} to the most general situation where both subsystems are described by symmetric operators, unfortunately a general proof is still missing.


\subsection{Separable dynamics and self-adjoint extensions}

We have seen so far that the space of self-adjoint extensions of a composite system is much bigger than the space of self-adjoint extensions of the subsystems. The remarkable fact is that even if the symmetric Hamiltonian operator is of the form in Equation \eqref{eq:sepHamiltonian} there are self-adjoint extensions that lead to non-separable dynamics, i.e. that entangle the subsystems. We are going to use the example above to show this phenomenon. Unfortunately, von Neumann's theorem is not suitable to perform explicit calculations. We are going to use the approach introduced in \cite{asorey05} and further developed in \cite{ibortlledo13}.

For doing so we take advantage of the following isomorphism $$\L^2([0,\infty))\hat{\otimes}\H_B\simeq \L^2([0,\infty);\H_B)\;.$$ In our example this isomorphism becomes $$\H=\L^2(\R^+;\mathbb{C}^2)$$ and we can consider that the elements in $\H$ are pairs 
$$\Phi=	\begin{bmatrix}
			\Phi_1\\\Phi_2
		\end{bmatrix}	
$$ and that the scalar product in $\H$ is given by 
$$\scalar{\Phi}{\Psi}=\int_{\R^+}\scalar{\Phi(x)}{\Psi(x)}_{\H_B}\d x\;.$$ As the manifold is one dimensional in this case and according to \cite{asorey05,ibortlledo13} the Hamiltonian $H$ is self-adjoint if and only if 
$$0=\scalar{\Phi}{H\Psi}-\scalar{H\Phi}{\Psi}=\scalar{\dot{\varphi}}{\psi}_{x=0}-\scalar{\varphi}{\dot{\psi}}_{x=0}\;,$$ where we use small greek letters to denote the restriction to the boundary and dotted small greek letters to denote the restriction of the normal derivative, $\Phi|_{x=0}=\varphi$\,, $-\frac{\d\Phi}{\d x}|_{x=0}=\dot{\varphi}$\,. The maximally isotropic subspaces of the boundary form are in one to one correspondence with the graphs of unitary operators $$U:\L^2(\{0\};\H_B)\to\L^2(\{0\};\H_B)\;.$$ More concretely, given a unitary operator $U\in\mathcal{U}(\L^2(\{0\};\H_B)) \simeq \mathcal{U}(2)$\,, the domain of a self-adjoint extension is characterized by those functions that satisfy the following boundary condition
\begin{equation}\label{eq:asorey1}
	\varphi+i\dot{\varphi}=U(\varphi-i\dot{\varphi})\;.
\end{equation}
Notice that $\varphi,\dot{\varphi}$ take values in $\H_B$\;.\\

We will consider unitary operators of the form $$U=U_A\otimes V\;.$$ As an example of separable dynamics we can take $U_A:\varphi\to e^{i\alpha} \varphi$ and $V=\1$\,. An easy calculation using $\{\rho_a\}_{a=1,2}$\,, the orthonormal base of $\H_B$\,, leads to a splitting of the evolution equation in two parts. Each one with the same boundary condition 
$$
\varphi_a+i\dot{\varphi}_a=e^{i\alpha} (\varphi_a-i\dot{\varphi}_a)\Leftrightarrow 
	\begin{cases}
		\dot{\varphi}_a=\tan\frac{\alpha}{2}\varphi_a & ,\alpha\neq\pi\\
		\varphi_a=0 & ,\alpha= \pi
	\end{cases}\;.
$$

In fact, this is a general result that can be applied to more general situations. A more general theorem that can be found in \cite[Theorem 2]{Ibort12} shows that separable dynamics can be achieved if and only if the boundary conditions are of the form $U=U_A\otimes\1$\,.


\subsection{Non-separable dynamics: An example}
Now we consider a unitary operator of the form $U=\1\otimes V$ with 
$$
V=	\begin{pmatrix}
		e^{i\alpha_1} & 0 \\ 0 & e^{i\alpha_2}
	\end{pmatrix}\;.
$$

Using again the decomposition provided by the orthonormal base of $\H_B$ we get in this case 
$$
H=H_A\otimes\1+\1\otimes H_B=
	\begin{pmatrix}
		-\frac{\d^2}{\d x^2}+\lambda_1 & 0 \\ 0 & -\frac{\d^2}{\d x^2}+\lambda_2
	\end{pmatrix}\;.
$$
The boundary condition  $\varphi+i\dot{\varphi}=U(\varphi-i\dot{\varphi})$ reads now
$$
\begin{bmatrix}
	\varphi_1+i\dot{\varphi}_1 \\ \varphi_2+i\dot{\varphi}_2
\end{bmatrix}=
\begin{pmatrix}
	e^{i\alpha_1} & 0 \\ 0 & e^{i\alpha_2}
\end{pmatrix}
\begin{bmatrix}
	\varphi_1-i\dot{\varphi}_1 \\ \varphi_2-i\dot{\varphi}_2
\end{bmatrix}\;,
$$
or in components $\dot{\varphi}_a=\tan\frac{\alpha_a}{2}\varphi_a$, $a=1,2$, $\alpha_a\neq\pi$\;.

The spectral problem for $H$ becomes the following system of differential equations
$$
	\begin{cases}
		-\frac{\d^2}{\d x^2}\Phi_1+\lambda_1\Phi_1=E\Phi_1&\\
		\dot{\varphi}_1=\tan \alpha_1/2 \varphi_1 &,\alpha_1\neq\pi
	\end{cases}\;,
$$
$$
	\begin{cases}
		-\frac{\d^2}{\d x^2}\Phi_2+\lambda_2\Phi_2=E\Phi_2&\\
		\dot{\varphi}_2=\tan \alpha_2/2 \varphi_2 &,\alpha_1\neq\pi
	\end{cases}\;.
$$
Since the manifold is not compact, in general there will be no solutions of the former system, i.e. $H$ will have no point spectrum in general. However, under the assumption that $\lambda_1>\lambda_2>E$ there exist square integrable solutions. Namely, 
$$
\Phi_1(x)=C_1e^{-\sqrt{\lambda_1-E}x}\;,
$$
$$
\Phi_2(x)=C_2e^{-\sqrt{\lambda_2-E}x}\;,
$$
for a fixed value of the energy that satisfies $$E=\lambda_1-\tan^2\frac{\alpha_1}{2}\;,$$ $$E=\lambda_2-\tan^2\frac{\alpha_2}{2}\;.$$ This imposes a compatibility condition for the existence of the eigenvalue:
\begin{equation}\label{eq:impliciteq}
	\sigma:=\lambda_1-\lambda_2=\tan^2\frac{\alpha_1}{2}-\tan^2\frac{\alpha_2}{2}\;.
\end{equation}
This implicit equation gives a family of possible self-adjoint extensions each of which will posses a unique eigenvalue. The implicit curves are plotted in Figure~\ref{fig:compatibility}.

Since for fixed boundary conditions there is only one eigenfunction, the evolution will be stationary provided that the initial state is the eigenfunction. However, we can consider that the boundary conditions are deformed adiabatically. For instance, consider the one parameter family 
$$U(s)=	\begin{pmatrix}
			e^{2 i s} & 0 \\ 0 & e^{2i s'}
		\end{pmatrix}\;,
$$
where $s'$ is uniquely determined by the value of $s$, the fixed value $\sigma$ and the implicit equation \eqref{eq:impliciteq}. Now, if the parameter $s$ is deformed adiabatically, the evolution of the system will be close to the unique eigenfunction of the system
\begin{equation}
	\Phi(t,x)=C_1e^{-\tan{s(t)}x}\otimes\rho_1+C_2e^{-\tan{s'(t)}x}\otimes\rho_2\;.
\end{equation}
Whenever one of the values $s$ or $s'$ vanishes the state becomes separable and it is non-separable in other case.
\begin{figure}[h]
	\center
	\includegraphics[width=9cm]{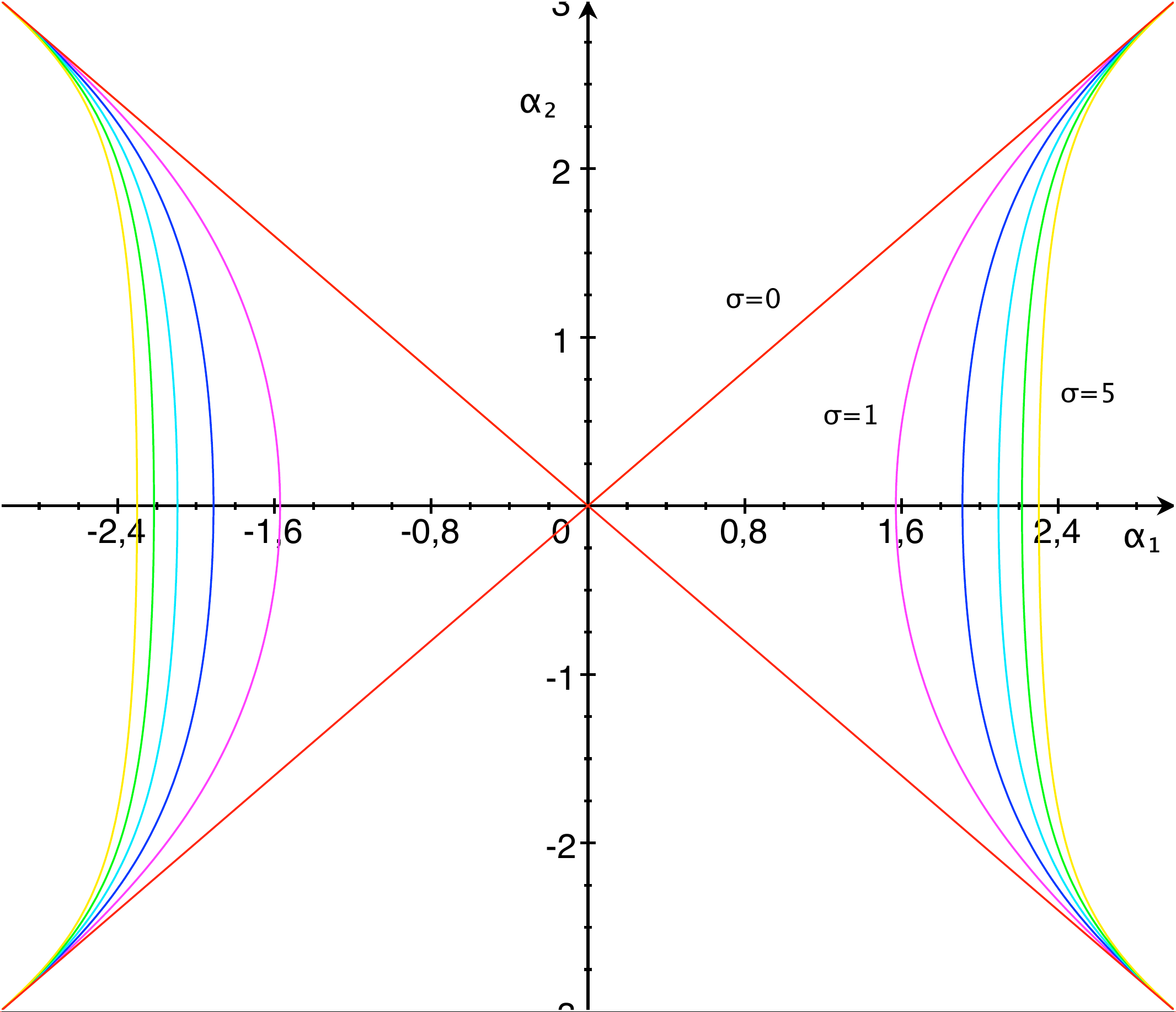}
	\caption{\footnotesize{Implicit curves for the parameters $\alpha_1$, $\alpha_2$ as a function of $\sigma=\lambda_1-\lambda_2$\,.}}\label{fig:compatibility}
\end{figure}
%


\newpage

\section{Lecture II.  Quadratic forms and self-adjoint extensions}\label{sec:qf}

In the previous lecture we have seen that although von Neumann's theory of self-adjoint extensions is exhaustive, it is not always the most suitable to characterize self-adjoint extensions.  In the particular case of differential operators the traditional way to describe self-adjoint extensions relies in the definition of appropriate boundary conditions. 

The aim of this lecture is to provide insight into this approach by using the Laplace-Beltrami operator as a guiding example.   We will discuss first the analytical difficulties that arise when dealing with this problem and then we will introduce the concept of closable quadratic forms to show how one can use them to avoid some of the difficulties.  

The self-adjoint extensions of the Laplace-Beltrami operator are analysed using this alternative approach.   It is noticeable that the use of quadratic forms in the study of linear operators has an extraordinary long and successful history running from applications to numerical analysis, like min-max methods, to the celebrated Weyl's formula for the asymptotic behaviour of eigenvalues of the Laplace-Beltrami operator, see \cite{LiebLoss} for an introduction and further references.   

In the recent paper \cite{Ibort13}, the ideas discussed in this lecture have been applied successfully to study the spectrum of Schr\"odinger operators in 1D and its stable and accurate numerical computation.

\subsection{Analytical difficulties} \label{sec:difficulties}

Through the examples in Lecture II (\S \ref{sec:bipartite}) we have seen that one can use the boundary term associated to Green's formula to characterize the self-adjoint extensions of differential operators. In particular we will focus on the Laplace-Beltrami operator but the following considerations can be applied to any formally self-adjoint differential operator. 

Consider a smooth Riemannian manifold $(\Omega,\eta)$ with smooth boundary $\pO$\,. We will only consider situations where $\pO\neq\emptyset$\,. Notice that the boundary of a Riemannian manifold has itself the structure of a Riemannian manifold if one considers as the Riemannian metric $\partial\eta=i^*\eta$\,, the pull-back under the canonical inclusion mapping $i:\pO\to\Omega$ of the metric $\eta$.  The Laplace-Beltrami operator on the Riemannian manifold $(\Omega,\eta)$ is the differential operator
\begin{equation}\label{eq:LB}
	\Delta_\eta\Phi=\frac{1}{\sqrt{|\eta|}}\frac{\partial}{\partial x^i}\sqrt{|\eta|}\eta^{ij}\frac{\partial\Phi}{\partial x^j}\;.
\end{equation}
For the Laplace-Beltrami operator Green's formula reads
$$\scalar{\Phi}{-\Delta\Psi}-\scalar{-\Delta\Phi}{\Psi}=\scalar{\dot{\varphi}}{\psi}_{\pO}-\scalar{\varphi}{\dot{\psi}}_{\pO}\;,$$
where as before we use small greek letters to denote the restriction to the boundary, $\varphi=\Phi|_{\pO}$, and dotted small greek letters to denote the restriction to the boundary of the normal derivative, $\dot{\varphi}=\d\Phi(\nu)|_{\pO}$\;, where $\nu$ is the normal vector field to the boundary pointing outwards. 

As already mentioned, the space of self-adjoint extensions is characterized by the space of maximally isotropic subspaces of the boundary form. It is known, cf. \cite{asorey05,ibortlledo13}, that maximally isotropic subspaces of this boundary form are in one-to-one correspondence with the set of unitary operators $\mathcal{U}(\L^2(\pO))$\;.   This correspondence is established in terms of the boundary equation:
\begin{equation}\label{eq:asorey}
	\varphi-i\dot{\varphi}=U(\varphi+i\dot{\varphi})\;.
\end{equation}
Hence, for each $U$, the equation above defines a self-adjoint domain for the Laplace-Beltrami operator. This equation has to be understood as an equality among vectors in $\L^2(\pO)$. Unfortunately, in dimension higher than one, the boundary data are not generic vectors of $\L^2(\pO)$\,. Recall that the Hilbert space $\L^2(\Omega)$ is defined to be a space of equivalence classes of functions.  Two functions define the same element in $\L^2(\Omega)$ if they differ only in a null measure set. In particular, this means that the restriction to the boundary of generic vectors $\Phi\in\L^2(\Omega)$ is ill defined since the boundary is a null measure set.  In order to overcome this difficulty we need to introduce the concept of Sobolev spaces. These spaces are the natural spaces to define the domains of differential operators.\\

Let $\beta\in\Lambda^1(\Omega)$ be a one-form on the Riemannian manifold $(\Omega,\eta)$\,. We say that $\mathrm{i}_V\beta=\beta(V)$ is a weak derivative of $\Phi\in\L^2(\Omega)$ in the direction $V\in\mathfrak{X}(\Omega)$ if 
$$
\int_\Omega\Phi(x)\mathrm{i}_V\d\Psi(x)\d\mu_\eta(x)=-\int_\Omega\mathrm{i}_V\beta(x)\Psi(x)\d\mu_\eta(x)\,,\quad\forall \Psi\in\C_c^\infty(\Omega)\;,
$$
where $\C^\infty_c(\Omega)$ is the space of smooth functions with compact support in the interior of $\Omega$\,. In such case we denote $\beta=\d\Phi$\,.   Select a locally finite covering $\{ U_\alpha \}$ of $\Omega$, with $U_\alpha$ small enough such that on each open set there is a well defined frame $\{e_i^{(\alpha)}\}\subset\mathfrak{X}(U_\alpha)$, and let $\rho_\alpha$ denote a partition of the unity subordinated to this covering.   The Sobolev space of order 1, $\H^1(\Omega)$, is defined to be 
\begin{align}
	\H^1(\Omega):= & \biggl\{\Phi\in\L^2(\Omega)\Bigr| \exists \beta=\d\Phi\in\L^2(\Omega)\otimes T^*\Omega \text{ such that}  \notag\\
	&\norm{\Phi}^2_1=\int_\Omega |\Phi(x)|^2\d\mu_\eta(x) + \sum_\alpha \rho_\alpha \sum_i\int_{U_\alpha} |\mathrm{i}_{e_i^{(\alpha)}} \d\Phi(x)|^2\d\mu_\eta(x)<\infty \biggr\}\;.\label{eq:H1}
\end{align}
The Sobolev spaces of higher order, $\H^k(\Omega)$, with norm $\norm{\cdot}_k$, are defined accordingly using higher order derivatives. An equivalent definition that is useful in manifolds without boundary is the following. The Sobolev space of order $k\in\mathbb{R^+}$ is defined to be 
$$
\H^k(\Omega):=\biggl\{ \Phi\in\L^2(\Omega) \Bigr| \norm{\Phi}^2_k=\int_\Omega\bar{\Phi}(x)(1-\Delta)^{k/2}\Phi(x)\d\mu_\eta(x)<\infty \biggr\}\;.
$$
Notice that there is no ambiguity in the use of the Laplace operator in this definition since in manifolds without boundary it possesses only one self-adjoint extension. This definition also works for non-integer values of $k$. Sobolev spaces of fractional order at the boundary are indeed the natural spaces for the boundary data. The following important theorem, cf. \cite{adams03,lions72}, establishes this correspondence.

\begin{theorem}[Lions' Trace Theorem]\label{thm:lions}
	The restriction map 
$$
	\map[b]{\C^\infty(\Omega)}{\C^\infty(\pO)}{\Phi}{\Phi|_{\pO}}
$$ can be extended to a continuous, surjective map 
$$
	b\,:\, \H^{k}(\Omega)\to\H^{k-1/2}(\pO)\,\quad k>1/2\;.
$$
\end{theorem}
Hence, even though elements of $\L^2(\Omega)$ do not possess well defined restrictions to the boundary, the functions in the Sobolev spaces do have them.

In the particular case of the Laplace-Beltrami operator, which is a second order differential operator, one natural choice for the domain would be the Sobolev space of order 2, $\H^2(\Omega)$. Lion's trace theorem, Theorem~\ref{thm:lions}, establishes that the boundary data necessarily belong to:
$$
(\varphi,\dot{\varphi})\in\H^{3/2}(\pO)\times\H^{1/2}(\pO)\;.
$$

Now we need to look closer to the boundary equation \eqref{eq:asorey}. We have that the boundary data are elements of some Sobolev spaces $\H^k(\partial\Omega)$ for some $k$.  However, the maximally isotropic subspaces are characterized in terms of unitary operators $\mathcal{U}(\L^2(\pO))$ and thus the right hand side of Eq. \eqref{eq:asorey} can be out of any $\H^k(\partial \Omega)$\;.   Hence, the boundary equation is not completely meaningful per se. The regularity of the boundary data plays a role even for simple boundary conditions like Robin boundary conditions. For instance, what happens if the Robin parameter is not continuous along the boundary? Such a situation appears in nature. Robin boundary conditions can model interphases between superconductors and insulators, \cite{asorey13}. A situation like the one shown in Figure \ref{fig:interphase} would be described by a discontinuous Robin parameter.

\begin{figure}
\centering
\includegraphics[width=8cm]{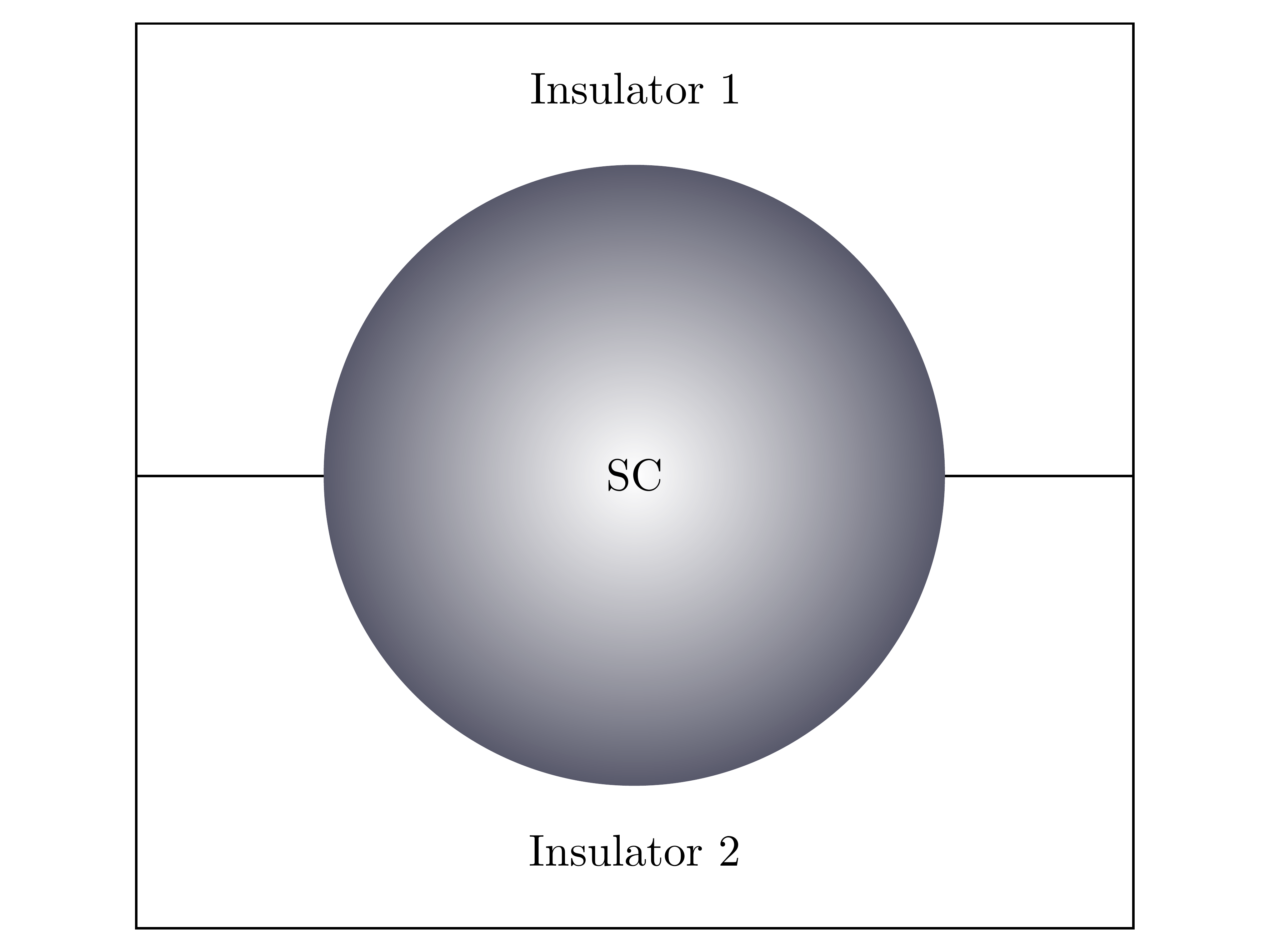}
\caption{\footnotesize{Superconductor surrounded by two different insulators. The boundary conditions between the superconductor and Insulator 1 are described by a constant parameter $\lambda_1$\,. Respectively for Insulator 2.}}\label{fig:interphase}
\end{figure}

There are several ways to handle this difficulty and each one gives rise to different theories or approaches to the problem of self-adjoint extensions. The theory developed by G. Grubb, \cite{grubb68}, expresses the boundary conditions in terms of a family of pseudo-differential operators. Hence the space of self-adjoint extensions is no longer parametrized as a unitary operator acting at the boundary. On the other hand the theory of boundary triples, \cite{bruning08}, keeps the structure of the boundary equation but expresses it in some abstract spaces. In this way the space of self-adjoint extensions can still be parametrized in terms of unitary operators but the relation with the boundary data becomes blurred. The approach that we shall pursue will rely on imposing appropriate conditions on the unitary operators $U\in\mathcal{U}(\L^2(\Omega))$\,. This allows us to keep both, the description in terms of unitary operators and the direct relationship with the space of boundary data. It will fail though to describe the set of all self-adjoint extensions. However, the class of self-adjoint extensions that we will describe in this way is wide enough to include all the well known boundary conditions.\\

Another complication that arises when one tries to describe self-adjoint extensions in terms of boundary conditions is that they might not be enough to characterize a self-adjoint extension. In particular, this can happen if the boundary is not smooth. We show this through the following example.

\begin{figure}[h]
\center
\includegraphics[width=8cm]{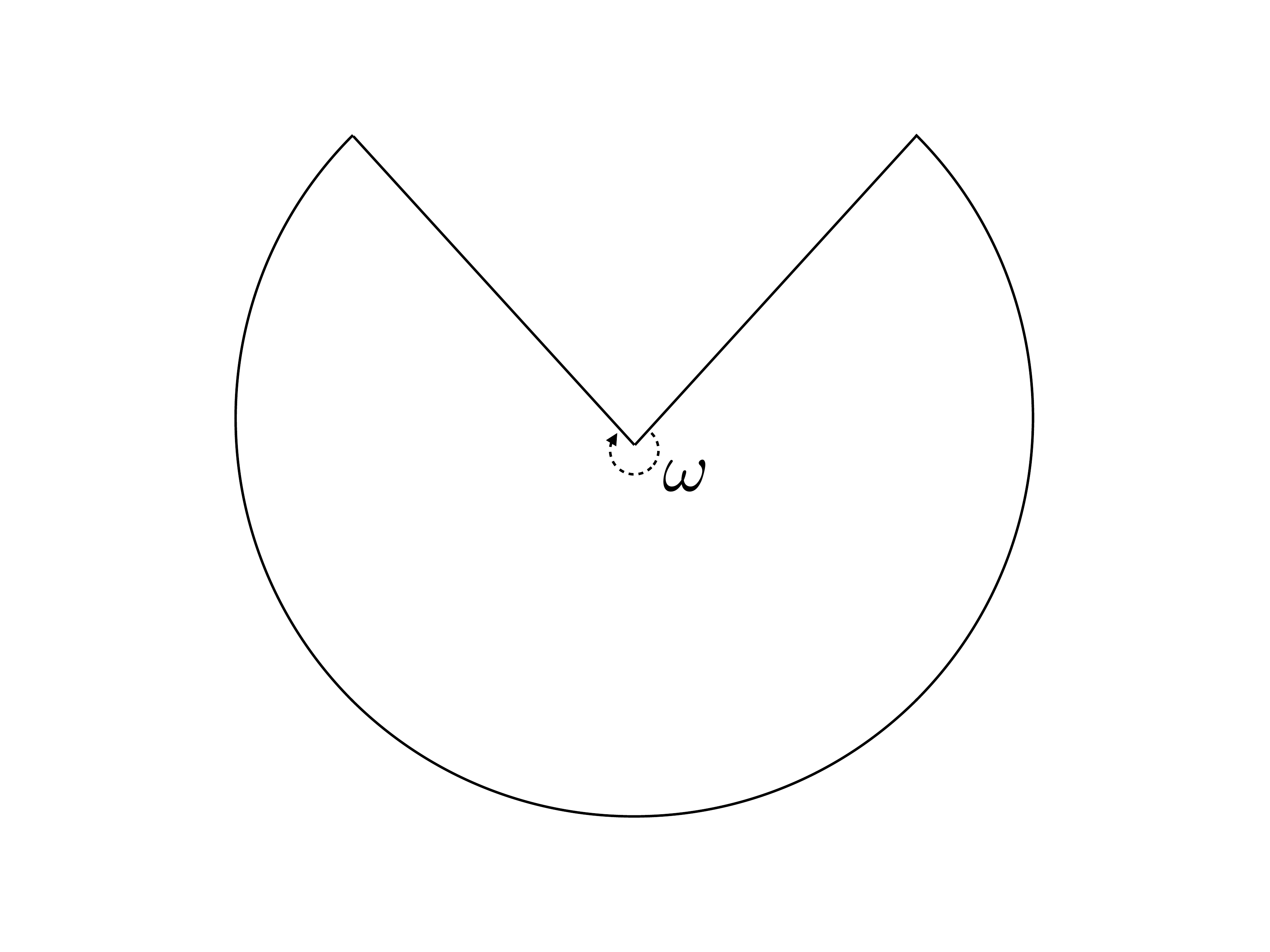}
\caption{\footnotesize{Disk of radius $1$ with a corner of angle $\omega$ removed.}}\label{fig:pizza}
\end{figure}

Consider the two dimensional Riemannian manifold obtained by removing a triangular slice from the unit disk such that the vertex of the triangle is at the centre of the disk, see Figure \ref{fig:pizza}. Now consider the Laplace-Beltrami operator defined on this manifold  with Dirichlet boundary conditions, i.e., $\varphi=\Phi|_{\pO}\equiv 0$\,. Then one can compute the dimension of the deficiency spaces in this case and obtain $$n_\pm=\dim\mathcal{N}_{\pm}=1\;.$$ This shows that the Laplace-Beltrami operator is not a self-adjoint operator in this situation. Another way to see this is to consider the following function in polar coordinates $$\Phi(r,\theta)=r^\beta\sin(\theta/\beta)\;.$$ This function satisfies Dirichlet boundary conditions if  $\beta=\frac{\omega}{\pi}$, where $\omega$ is the angle of the triangular slice removed form the disk. One can check that $\Phi\not\in\H^2(\Omega)$ and hence it is not an element of the domain of the Laplace-Beltrami operator $\D(\Delta)$. However it is easy to verify that $\Delta \Phi \equiv0$ and hence $\Phi\in\D(\Delta^\dagger)$.\\

The latter example shows that the domain of the adjoint operator does strongly depend on the regularity properties of the boundary manifold. In general it is very hard to obtain this domain explicitly. Unfortunately all the approaches mentioned so far to describe self-adjoint extensions depend on the knowledge of the domain of the adjoint operator.

In general, to use boundary conditions to find maximally isotropic subspaces of the boundary form only guarantees that one finds symmetric extensions of the corresponding operators. In order to ensure that those boundary conditions indeed select self-adjoint domains one needs to do further analysis. 

In the next section we shall introduce one of the main tools to describe self-adjoint extensions, namely Kato's Representation Theorem. We will use it to present an alternative way to characterize self-adjoint extensions of differential operators in terms of boundary conditions that overcomes many of the difficulties presented so far. 

\subsection{Kato's Theorem and Friedrichs' Extensions}

There is an alternative approach to von Neumann's Theorem that allows for the characterization of self-adjoint extensions of differential operators. This approach is based on the use of quadratic forms and allows to characterize uniquely self-adjoint extensions in terms of boundary conditions. We need some definitions first.
\begin{enumerate}
	\item Consider a dense subset $\D\subset\H$. A quadratic form on $\D$ is the evaluation on the diagonal of a hermitean, sesquilinear form, i.e. $Q:\D\times\D\to\mathbb{C}$ such that 
		$$
			Q\pair{\Phi}{\Psi}=\overline{Q\pair{\Psi}{\Phi}}\;.
		$$
		We will use the same symbol to denote the quadratic form and the sesquilinear form 
		$$
			\map[Q]{\D}{\R}{\Phi}{Q\pair{\Phi}{\Phi}}\;.
		$$
	\item We say that a quadratic form is semibounded from below  if there exists a constant $a>0$ such that 
		$$
			Q(\Phi)>-a\norm{\Phi}^2\,,\quad\forall\Phi\in\D\;.
		$$
		Equivalently, one can define semibounded from above. $Q$ is semibounded from above if there exists $a>0$ such that 
		$$
			Q(\Phi)<a\norm{\Phi}^2\,,\quad\forall\Phi\in\D\;.
		$$
	\item A semibounded quadratic form $Q$ is closed if its domain is closed with respect to the graph norm defined by
		$$
			\normm{\Phi}^2_Q=(1+a)\norm{\Phi}^2  \pm Q(\Phi)\;.
		$$
		The different signs correspond to the quadratic form being semibounded either from below or from above. Notice that the factor $(1+a)$ guarantees that the expression above is a norm.
		Even if the quadratic form is not closed one can consider the closure of its domain with respect to this norm, i.e., $\overline{\D}^{\smash{\normm{\cdot}_Q}}$\,. Unfortunately the quadratic form can not always be continuously extended to this domain. When the quadratic form can be continuously extended we say that the quadratic form is closable and we denote the closed extension of the quadratic form in the domain $\overline{\D}^{\smash{\normm{\cdot}_Q}}$ by $\overline{Q}$\,. Closability is strictly related with the continuity properties of the quadratic form on its domain. In this way it can tell us when can one interchange the limits 
		$$
			\lim_{n\to\infty}Q(\Psi_n)=Q(\lim_{n\to\infty}\Psi_n)\;.
		$$
\end{enumerate}

We are ready to enunciate the most important result of this section. For the proof we refer to \cite{adams03,kato95,reed75}.

\begin{theorem}[Kato's representation theorem]\label{thm:kato}
	Let $Q$ be a closed, semibounded quadratic form with domain $\D$. Then it exists a unique, self-adjoint, semibounded operator $T$ with domain $\D(T)\subset\D$ such that 
	$$
		Q\pair{\Phi}{\Psi}=\scalar{\Phi}{T\Psi}\,\quad\forall\Phi\in\D\,,\forall\Psi\in\D(T)\;.
	$$
\end{theorem}

As an application we will obtain the domain of the Neumann extension of the Laplace-Beltrami operator.

\begin{example}\label{ex:quadneuman}
Let $\Omega$ be a Riemannian manifold and consider the following quadratic form defined in $\L^2(\Omega)\,.$
$$
	Q(\Phi,\Psi)=\sum_i\scalar{\frac{\partial\Phi}{\partial x_i}}{\frac{\partial\Psi}{\partial x_i}}\,,\quad \Phi,\Psi\in\D=\H^1(\Omega)\;.
$$
Recall that the scalar product in $\H^1(\Omega)$, see Equation \eqref{eq:H1}, satisfies that 
$$
	\norm{\Phi}_1=\norm{\Phi}+Q(\Phi)=\normm{\Phi}_Q\,,
$$
and hence it is immediate that the quadratic form above is closed. According to Kato's theorem there exists a self-adjoint operator $T$ that represents the quadratic form. This self-adjoint operator turns out to be the the Neumann extension of the Laplace-Beltrami operator. Notice that so far we have not imposed any boundary condition, so where is the condition $\dot{\varphi}=0$ coming from? In order to obtain this condition we need the following characterization of the domain of the operator $T$\,. The proof can be found at \cite{davies95}.

\begin{lemma}\label{lem:characterizationdavies}
Let $Q$ be a closed, quadratic form with domain $\D$ and let $T$ be the representing operator. An element $\Psi\in\H$ is in the domain of the operator $T$, i.e. $\Psi\in\D(T)$\,, if and only if $\Psi\in\D$ and there exists $\chi\in\H$ such that for all $\Phi\in\D$
$$
	Q(\Phi,\Psi)=\scalar{\Phi}{\chi}\;.
$$
In such a case one can define $\chi:=T\Psi$\,.
\end{lemma}

We want to apply this characterization to our case. If we use Green's formula we get
$$Q(\Phi,\Psi)=\scalar{\Phi}{-\Delta\Psi}+\scalar{\varphi}{\dot{\psi}}\stackrel{!}{=}\scalar{\Phi}{\chi}\;$$
and hence in order for the last equality to hold we need to impose two extra conditions. First, for the first summand to be defined we need $\Psi\in\H^2(\Omega)$, since the Laplace operator is a second order differential operator. Second, the boundary term has to vanish. Since it has to vanish for all $\Phi\in\H^1$ and there is no restriction on the possible boundary values that $\Phi$ can take, this forces $\dot{\psi}=0$. Hence, if $\Psi\in\H^2(\Omega)$ and $\dot{\psi}=0$ we have that $\Psi\in\D(T)$ and moreover $T\Psi=\chi=-\Delta\Psi$\,, i.e. $T$ is the Neumann extension of the Laplace-Beltrami operator.
\end{example}

Our main aim is to introduce tools that allow us to characterize self-adjoint extensions in an unambiguous manner. Kato's representation theorem above allows us to characterize uniquely self-adjoint operators once we have found an appropriate closed or closable quadratic form. It is worth to notice that in order to obtain such operator there is no need to know the domain of the adjoint operator. The following result, again without proof, allows to characterize uniquely a self-adjoint extension of a given symmetric operator.

\begin{theorem}[Friedrichs extension theorem]\label{thm:friedrichs}
	Let $T_0$ be a symmetric, semibounded operator with domain $\D(T_0)$\,. Then the quadratic form 
	$$
		Q_{T_0}(\Phi,\Psi):=\scalar{\Phi}{T_0\Psi}\,,\quad\Phi,\Psi\in\D(T_0)
	$$
	 is closable.
\end{theorem}

The above result is applied as follows. Given a symmetric, semibounded operator $T_0$ one can construct the quadratic form $\scalar{\Phi}{T_0\Psi}$\,, which is closable. Now one can consider its closure and apply Kato's representation theorem to characterize a unique self-adjoint operator $T$ that extends $T_0$. 
In the spirit of the approach using boundary conditions it means that it is enough to characterize boundary conditions that guarantee that the operator $T_0$ is symmetric. Its Friedrichs extension will then be a self-adjoint extension of it with the postulated boundary conditions.

\subsection{The Laplace-Beltrami operator revisited}\label{sec:LBrevisited}

As we have seen, regardless of the analytical difficulties mentioned in Section \ref{sec:difficulties}, one can use the boundary equation \eqref{eq:asorey}, i.e.
$$
	\varphi-i\dot{\varphi}=U(\varphi+i\dot{\varphi})\;,
$$
 to define easily symmetric extensions of the Laplace-Beltrami operator \eqref{eq:LB}. In general for a generic $U\in\mathcal{U}(\L^2(\pO))$ the boundary conditions above will not describe a self-adjoint domain for the Laplace-Beltrami operator, however they always define symmetric domains for it. Nevertheless we can use Friedrichs' theorem, Theorem \ref{thm:friedrichs}, to characterize uniquely a self-adjoint extension. As we have seen we just need to ensure the semiboundedness of the corresponding symmetric operator.\\

We consider as symmetric operator the Laplace-Beltrami operator \eqref{eq:LB} defined on the domain
$$
	\D_U=\left\{  \Phi\in\H^2(\Omega)\mid \varphi-i\dot{\varphi}=U(\varphi+i\dot{\varphi})  \right\}\;.
$$
Since $U\in\mathcal{U}(\L^2(\pO))$ does not verify any special condition it may very well be that the boundary equation has only the trivial solution 
$$
	(\varphi,\dot{\varphi})=(0,0)\in\H^{3/2}(\pO)\times\H^{1/2}(\pO)\;.
$$
However, this condition does define a symmetric extension of the Laplace-Beltrami operator. Following the lines of Example~\ref{ex:quadneuman} at the end of the previous section it is easy to show that the Friedrichs extension associated to this symmetric operator is precisely the Dirichlet extension of the Laplace-Beltrami operator. We leave this as an exercise.

In order to apply Friedrichs' extension theorem to generic unitary operators at the boundary we need to ensure that the associated quadratic forms are semibounded. Using Green's formula for the Laplace operator once we get the following expression,
\begin{equation}\label{eq:greenLB}
	\scalar{\Phi}{-\Delta\Phi}=\scalar{\d\Phi}{\d\Phi}-\scalar{\varphi}{\dot{\varphi}}\;.
\end{equation}
The first summand at the right hand side is automatically positive. Hence, in order to analyse the semiboundedness of the Laplace-Beltrami operator it is enough to analyse the semiboundedness of the boundary term, i.e.
$$
	-\scalar{\varphi}{\dot{\varphi}}\stackrel{?}{\geq} a\norm{\Phi}^2
$$
for some $a\in \R$\,. 

In general this term does not satisfy the bound above for any constant $a\in\mathbb{R}$\,. However, under certain conditions on the spectrum of the unitary operator $U$ it is possible to show, cf. \cite{ibortlledo13}, that it is indeed semibounded. Showing this in detail would take us away from the scope of this introductory approach, so we will just sketch the proof. Any further details can be found at \cite{ibortlledo13}.

The main idea is to use the boundary equation in order to express $\dot{\varphi}$ as a function of $\varphi$. Solving for $\dot{\varphi}$ at the boundary equation \eqref{eq:asorey} we get formally that 
\begin{equation}\label{eq:dotphiproptophi}
	\dot{\varphi}=i\frac{U-\mathbb{I}}{U+\mathbb{I}}\varphi:=A_U\varphi\;.
\end{equation}
The expression $$i\frac{U-\mathbb{I}}{U+\mathbb{I}}=A_U\;,$$ which is known as the Cayley transform of the unitary operator $U$, defines a linear, self-adjoint operator $A_U$. However, as long as $-1$ is in the spectrum of the unitary operator $U$, which we denote $-1\in\sigma(U)$, this operator is not a bounded operator. We shall impose conditions on the unitary operator $U$ that guarantee that the relation \eqref{eq:dotphiproptophi} is given by a bounded operator.

\begin{definition}\label{def:gap}
We will say the the unitary operator $U\in\mathcal{U}(\L^2(\pO))$ has gap if there exists  a $\delta>0$ such that 
$$\sigma(U)\backslash\{-1\} \subset \{e^{i\theta} \mid -\pi+\delta \leq \theta \leq \pi-\delta\}\;.$$ 
\end{definition}
In other words, there must be a positive distance between $-1\in\mathbb{C}$ and the closest element of the spectrum different of $-1$, cf. Figure \ref{fig:spectrum}.\\

\begin{figure}[h]
\center
\includegraphics[width=8cm]{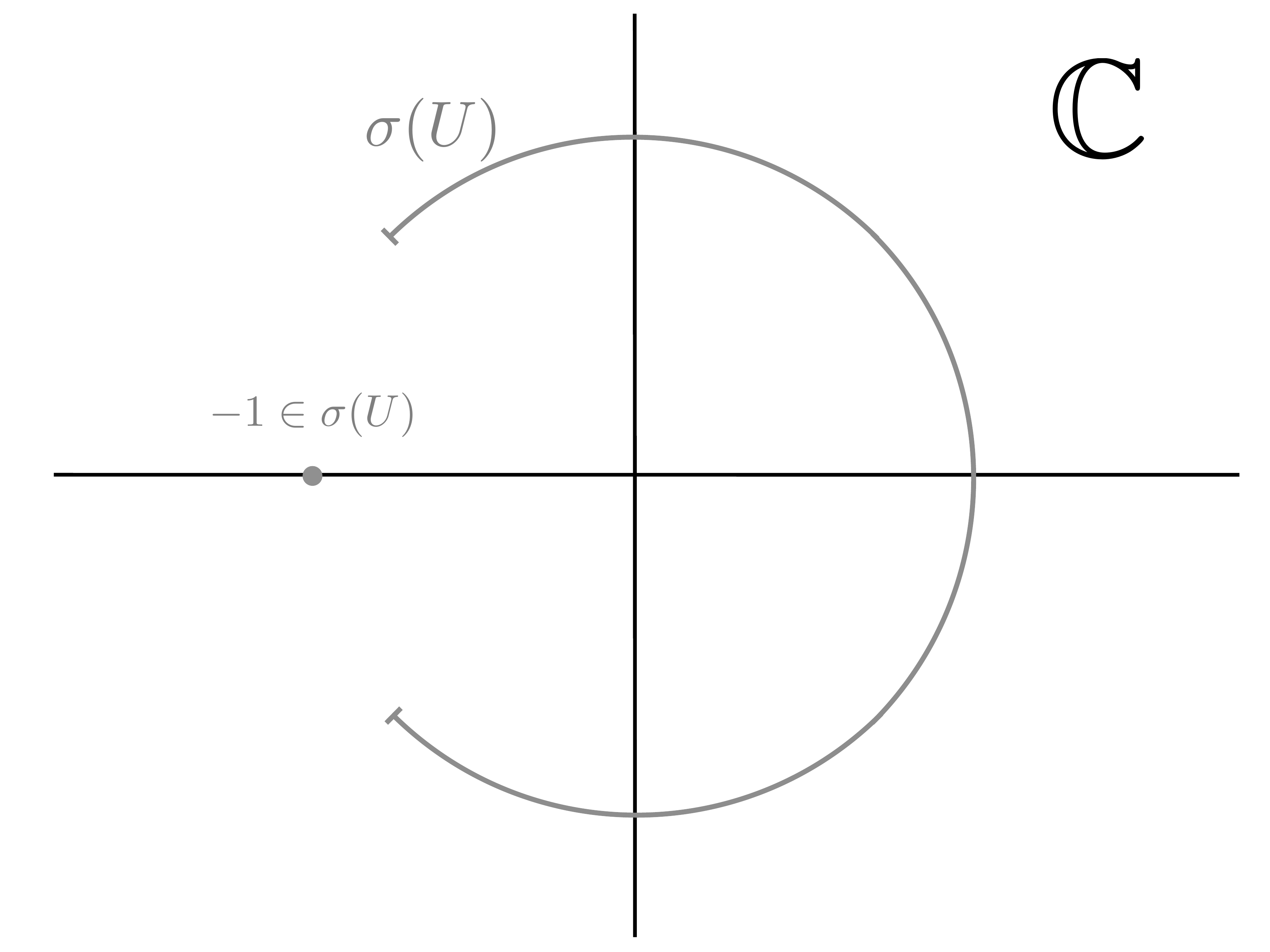}
\caption{\footnotesize{Spectrum of a unitary operator $U\in\mathcal{U}(\L^2(\pO))$ with gap.}}\label{fig:spectrum}
\end{figure}

We define the subspace $W\subset\L^2(\pO)$ to be the proper subspace associated to the eigenvalue $-1$. Notice that this subspace can be reduced to the case $W=\{0\}$ if $-1\notin\sigma(U)$\,. The Hilbert space at the boundary can be split as $$\L^2(\pO)=W \oplus W^{\bot}\;.$$ According to this decomposition we can rewrite the boundary equation as 
\begin{subequations}\label{eq:bcqf}
\begin{equation}
	(\mathbb{I}+U)\dot{\varphi}_W=-i(\mathbb{I}-U)\varphi_W\;\Rightarrow \varphi_W=0\;,
\end{equation}
\begin{equation}
	(\mathbb{I}+U)\dot{\varphi}_{W^\bot}=-i(\mathbb{I}-U)\varphi_{W^\bot}\;\Leftrightarrow \dot{\varphi}_{W^\bot}=A_{U_{W^\bot}}\varphi_{W^\bot}\;.
\end{equation}
\end{subequations}
Notice that $A_{U_{W^\bot}}$ is a bounded operator. Hence, the boundary term satisfies
$$
|\scalar{\varphi}{\dot{\varphi}}|=|\scalar{\varphi_{W}}{\dot{\varphi}_{W}}+\scalar{\varphi_{W^\bot}}{\dot{\varphi}_{W^\bot}}|=|\scalar{\varphi_{W^\bot}}{A_{U_{W^\bot}}\varphi_{W^\bot}}|\leq \norm{A_{U_{W^\bot}}} \norm{\varphi}^2\;.
$$
With this condition, Green's formula for the Laplace-Beltrami operator, Eq.~\eqref{eq:greenLB} can be bounded by
$$
	\scalar{\Phi}{-\Delta\Phi}\geq\scalar{\d\Phi}{\d\Phi}-K\norm{\varphi}^2\;.
$$
The latter is the quadratic form associated to Robin boundary conditions with constant parameter $K=\norm{A_{U_{W^\bot}}}$\,. If we are able to show that the Robin extension of the Laplace-Beltrami operator is bounded from below for any value of $K>0$ we will show that any extension determined by a unitary with gap will also be semibounded from below. The proof of the latter can be split in several steps which we sketch below.

\begin{enumerate}
	\item The one dimensional problem $-\frac{\d^2}{\d r^2}\phi(r)=\Lambda \phi(r)$ with Robin boundary conditions is semibounded from below. We define the operator $R:=-\frac{\d^2}{\d r^2}$\,.
	\item One can always choose a coordinate neighbourhood of the boundary such that our problem becomes 
	$$-\Delta\simeq R\otimes\mathbb{I}-\mathbb{I}\otimes\Delta_{\pO}\;.$$
	The Laplace-Beltrami operator at the boundary, $-\Delta_{\pO}$, is always positive defined and therefore the lower bound of $-\Delta$ with Robin boundary conditions is that of $R$\,.
	\item By introducing the coordinate neighbourhood at the boundary one is introducing an auxiliary boundary. One needs to introduce boundary conditions there that do not spoil the bounds above. Since it is enough to bound the associated quadratic forms in order to bound the corresponding operators it is enough to consider Neumann boundary conditions. Remember Example~\ref{ex:quadneuman} where we showed that, at the level of quadratic forms, selecting Neumann boundary conditions amounts to select no boundary conditions at all.
\end{enumerate}

We have showed that unitaries with gap lead to lower semibounded, symmetric extensions of the Laplace-Beltrami operator. Hence, together with Friedrichs' Extension Theorem, each boundary condition of this form leads to a unique, self-adjoint extension of the Laplace-Beltrami operator.


\newpage

\section{Lecture III. Self-adjoint extensions of Dirac operators: First steps}\label{sec:Dirac}

To define the Laplace-Beltrami operator it is enough to have the structure of a smooth Riemannian manifold. However, in order to define Dirac operators one needs additional structures. In what follows we will briefly summarize the main notions that we shall need for the rest of this article.

\begin{alphenumerate}
	\item \label{alph:1} $(\Omega,\eta)$ is a smooth Riemannian manifold with smooth boundary $\pO$. As before we will only consider situations where $\pO\neq \emptyset$\,.
	\item The Clifford algebra $\mathrm{Cl}(T_x\Omega)$ generated by the tangent space $T_x\Omega$, $x\in\Omega$, is the associative algebra generated by $u\in T_x\Omega$ and such that 
	$$u\cdot v + v \cdot u = -2\eta_x(u,v)\,,\quad u,v\in T_x\Omega\;.$$
	The Clifford bundle over $\Omega$, denoted $\mathrm{Cl}(\Omega)$\,, is defined to be 
	$$\mathrm{Cl}(\Omega)=\bigcup_{x\in\Omega}\mathrm{Cl}(T_x\Omega)\;.$$
	\item A $\mathrm{Cl}(\Omega)-$bundle over $\Omega$ is a Riemannian vector bundle $$\pi:S\to\Omega$$ such that for all $x\in\Omega$ the fibre $S_x=\pi^{-1}(x)$ is a $\mathrm{Cl}(T_x\Omega)-$module, i.e. it exists a representation of the Clifford algebra $\gamma:\mathrm{Cl}(T_x\Omega) \to \mathrm{Hom}(S_x,S_x)$ such that for $u\in T_x\Omega$\,, $\xi \in S_x$
$$
\map[\gamma]{\mathrm{Cl}(T_x\Omega)\times S_x}{S_x}{(u,\xi)}{\gamma(u)\xi}\;.
$$
	For simplicity of the notation we will omit the explicit use of the representation of the Clifford algebra and write directly
$$
\map[u]{S_x}{S_x}{\xi}{u\cdot\xi}\;.
$$
	The left action of the Clifford algebra on the fibres of the Clifford bundle is called Clifford multiplication.
	\item Clifford multiplication by unit vectors acts unitarily with respect to the Hermitean structure $\pair{\cdot}{\cdot}_x$ of the Riemannian vector bundle $\pi:S\to\Omega$\,, i.e. for $u\in T_x\Omega$, such that $\eta_x(u,u)=1$
	$$
		\pair{u\cdot\xi}{u\cdot\zeta}_x=\pair{\xi}{\zeta}_x\quad \forall \xi,\zeta\in S_x\;.
	$$
	\item  A Hermitean connection on the vector bundle $\pi:S\to\Omega$, is a mapping $$\nabla: T_x\Omega \times \Gamma^\infty(S) \to \Gamma^\infty(S)$$ such that
	$$\pair{\nabla\xi}{\zeta}+\pair{\xi}{\nabla\zeta}=\d\pair{\xi}{\zeta}\;,$$ 
	or equivalently for $X\in \Gamma(T\Omega)=\mathfrak{X}(\Omega)$
	$$\pair{\nabla_{X}\xi}{\zeta}+\pair{\xi}{\nabla_X\zeta}=X\pair{\xi}{\zeta}\;.$$
	A Hermitean connection on the $\mathrm{Cl}(\Omega)-$bundle is a Hermitean connection compatible with the Levi-Civita connection $\nabla_{\eta}$ of the underlying Riemannian manifold $(\Omega,\eta)$, i.e. for $u\in \mathrm{Cl}(\Omega)$, $\xi\in \Gamma^\infty(S)$
	$$\nabla(u\cdot\xi)=\nabla_\eta u\cdot \xi + u\cdot\nabla\xi\;.$$
	\item Let $\nabla$ be an Hermitean connection on the $\mathrm{Cl}(\Omega)-$bundle and let $\{e_i\}\subset \Gamma^\infty(S)$ be an orthonormal frame. The Dirac operator $D:\Gamma^\infty(S)\to\Gamma^\infty(S)$ is the first order differential operator defined by 
	$$D\xi=\sum_i e_i\cdot\nabla_{e_i}\xi\;.$$
\end{alphenumerate}

We will illustrate the above construction with two simple examples.

\begin{example}
As an example where $\dim \Omega= 1$ we can take $\Omega=S^1$. The Clifford algebra of a one dimensional real vector space is isomorphic to the complex numbers. In this case one can take the Clifford bundle to be the complexified tangent space $TS^1\otimes\mathbb{C}$ and Clifford multiplication can be represented as complex multiplication. As orthonormal frame $\{e_1\}$ one can take the complex unit $i$ and therefore the Dirac operator becomes in this case $$D=i\frac{\partial}{\partial \theta}\;,$$
where $\theta$ is the coordinate in $S^1$\,.
\end{example}

\begin{example}\label{ex:diracstructureexample2}
We will consider now a two dimensional example. Consider that $\Sigma$ is a compact, orientable Riemannian surface with smooth boundary such that it has $r$ connected components $\partial\Sigma=\bigcup_{\alpha= 1}^r S_\alpha$, $S_\alpha\simeq S^1$\;. In Figure 2 there is an example with $r=2$\,. 

\begin{figure}[h]
\center
\includegraphics[width=8cm]{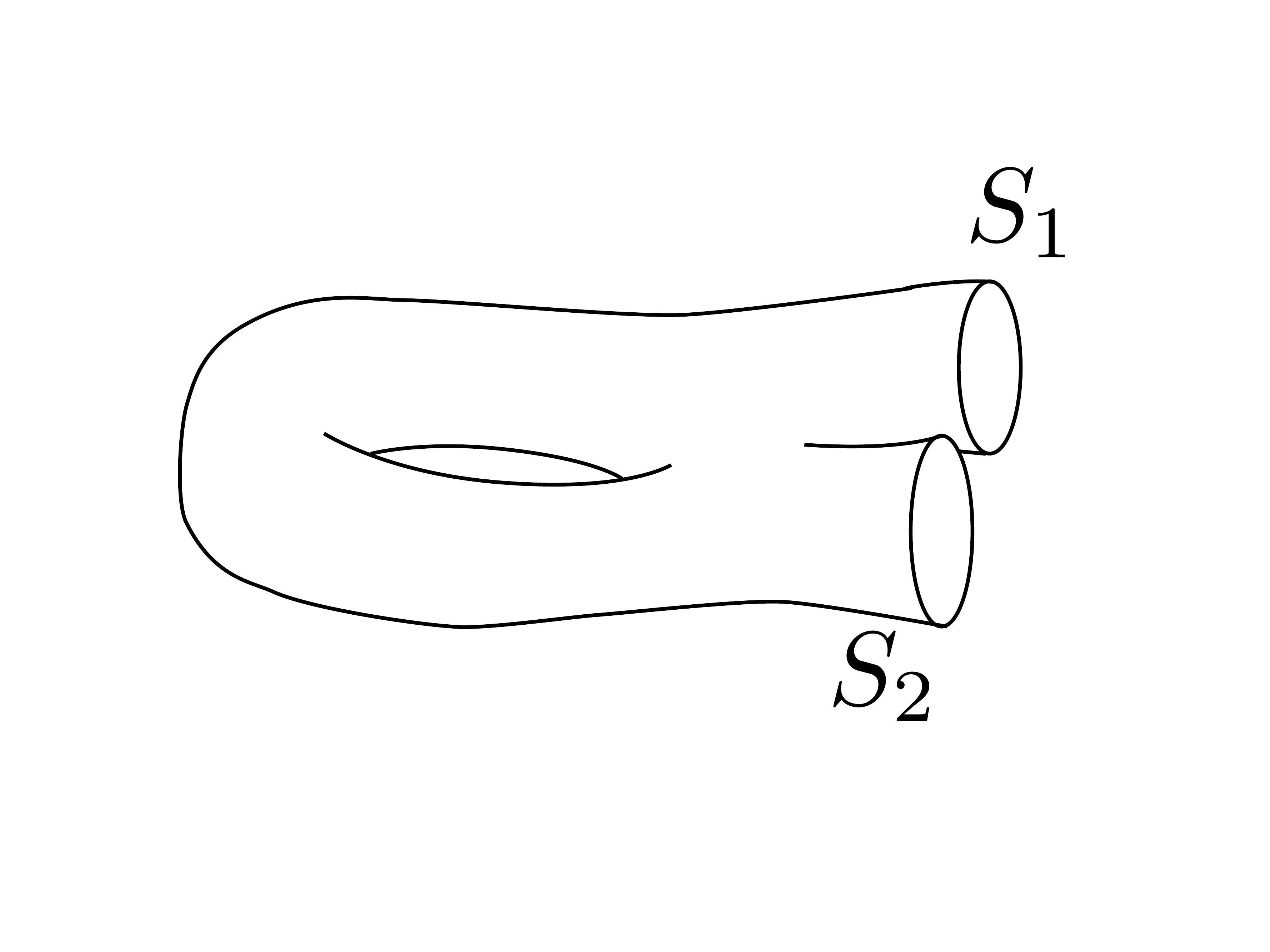}
\caption{\footnotesize{Riemannian surface with 2 connected components at the boundary}}
\end{figure}
For each $p\in\Sigma$ there is an open neighbourhood $U$ and a holomorphic transformation $\varphi$ such that 
$$
\map[\varphi]{U}{\mathrlap{\mathbb{C}}}{q}{\mathrlap{z=x+iy}}
$$
We can take now the orthonormal base of $T_p\Sigma$ generated by the coordinate vectors $e_1=\frac{\partial}{\partial x}$, $e_2=\frac{\partial}{\partial y}$ with $\scalar{\frac{\partial}{\partial x}}{\frac{\partial}{\partial y}}_p=0$\,. The Clifford algebra $\mathrm{Cl}(T_p\Sigma)$ is generated by the linear span of the elements $e_1, e_2, e_1\cdot e_2$\;, that satisfy the following relations: 
$$
e_1\cdot e_2+e_2\cdot e_1=0\,,\quad e_1^2=-1\,,\quad e_2^2=-1\;.
$$
This algebra can be represented in the algebra of 2 by 2 complex matrices. Take 
$$\sigma_1=\begin{bmatrix} 0 & i \\ i & 0  \end{bmatrix}\, ,\qquad
\sigma_2=\begin{bmatrix} 0 & -1 \\ 1 & 0  \end{bmatrix}\;,$$
which is the spin representation of the group $\mathcal{SO}(2)$ with spin group 
$$
	\map{\mathcal{U}(1)}{\mathcal{SO}(2)}{z}{z^2}\;.
$$

	Then the complexified tangent bundle $T\Sigma\otimes \mathbb{C}=T\Sigma^{\mathbb{C}}$ is a $\mathrm{Cl}(\Sigma)-$bundle and for $ \mathfrak{X}(\Sigma)\ni u=x e_1 + y e_2$ Clifford multiplication becomes 
\begin{align*}\notag
u\cdot\xi=(x e_1 + y e_2)\cdot \xi &= (x\sigma_1 + y \sigma_2)\xi\\
	&=\left(x \begin{bmatrix} 0 & i \\ i & 0  \end{bmatrix} + y \begin{bmatrix} 0 & -1 \\ 1 & 0  \end{bmatrix}\right)
	\begin{pmatrix} \xi_2 \\ \xi_1 \end{pmatrix}\\
	&=\begin{pmatrix} (ix - y)\xi_2 \\ (ix+y)\xi_1 \end{pmatrix}=
	\begin{pmatrix} iz \xi_2 \\ i\bar{z}\xi_1 \end{pmatrix}\;.
\end{align*}
Then, in geodesic coordinates around $p$ the Levi-Civita connection is $\nabla_{\eta}=\frac{\partial}{\partial x}\d x+\frac{\partial}{\partial y}\d y$ and taking it as the Hermitean connection on $T\Sigma^{\mathbb{C}}$ the Dirac operator becomes
\begin{align*}
	D&	=e_1\cdot \nabla_{e_1}+e_2\cdot \nabla_{e_2}\\ 
		&=\begin{bmatrix} 0 & i \\ i & 0  \end{bmatrix}\frac{\partial}{\partial x} + \begin{bmatrix} 0 & -1 \\ 1 & 0  \end{bmatrix}\frac{\partial}{\partial y}\\
		&=\begin{bmatrix} 0 & i\frac{\partial}{\partial x} - \frac{\partial}{\partial y} \\ i\frac{\partial}{\partial x}+\frac{\partial}{\partial y} & 0 \end{bmatrix}=\begin{bmatrix} 0 & i\partial \\ i\bar{\partial} & 0 \end{bmatrix}\;.
\end{align*}
As it should be we have that 
$$
D^2=\begin{bmatrix} -\partial\bar{\partial} & 0 \\ 0 & -\bar{\partial}\partial \end{bmatrix}=\nabla^\dagger\nabla + \frac{\kappa}{4}\;,$$
which is known as the Lichnerowicz identity. The first term at the right hand side is the connection Laplacian and $\kappa$ is the scalar curvature. An easy consequence of this equality is that for manifolds without boundary and with positive scalar curvature there are no Harmonic spinors. It is enough to realize that 
$$\norm{D\xi}^2=\int_{\Sigma}\norm{\nabla\xi}^2\d\mu_{\eta}+\frac{1}{4}\int_{\Sigma}\norm{\xi}^2\kappa\d\mu_{\eta}\;.$$
\end{example}


\subsection{Self-adjoint extensions of Dirac operators}

Once we have defined the Dirac operator as a first order differential operator on a Riemannian vector bundle we are interested in its properties as a linear operator acting on the Hilbert space of square integrable sections of the vector bundle $\pi:S\to\Omega$\,. The space of square integrable sections $\L^2(S)$ is defined to be the completion of the space of smooth sections $\Gamma^\infty(S)$ with respect to the norm induced by the scalar product 
$$
\scalar{\xi}{\zeta}=\int_{\Omega}\pair{\xi(x)}{\zeta(x)}_x\d\mu_{\eta}(x)\;,
$$
where $\d\mu_{\eta}$ is Riemannian volume form induced by the Riemannian metric $\eta$\,.

As the Laplace-Beltrami operator, and any differential operator, the Dirac operator is an unbounded operator and thus it needs a domain of definition. As we already mentioned, the natural subspaces for the definition of differential operators are the Sobolev spaces, which where introduced in the previous section. Since we are dealing now with vector bundles the definitions of Sobolev spaces change slightly and we introduce them in an independent manner.

Let $\beta\in\Lambda^{1}(\Omega;S)$ be a one-form on the Riemannian manifold $(\Omega,\eta)$ with values on the vector bundle $S$. We say that $\beta$ is a weak covariant derivative of $\xi\in\L^2(S)$ if 
$$
\int_\Omega\pair{\xi}{\nabla_V\zeta}_x\d\mu_\eta(x)=-\int_\Omega\pair{\mathrm{i}_V\beta}{\zeta}\d\mu_\eta(x)\,,\quad\forall \zeta\in\Gamma_c^\infty(S),\;\forall V\in\mathfrak{X}(\Omega)\;,
$$
where $\Gamma^\infty_c(S)$ is the space of smooth sections with compact support in the interior of $\Omega$\,. In such case we denote $\beta=\nabla\xi$\,.

The Sobolev space of sections of order $1$, $\H^1(S)$, is defined to be 
\begin{align*}
	\H^1(S):=\biggl\{\xi\in\L^2(S)\mid \exists \beta&=\nabla\xi\in\L^2(S)\otimes T^*\Omega \text{ and }  \\
	&\norm{\xi}^2_1=\int_\Omega \pair{\xi}{\xi}_x\d\mu_\eta(x) + \int_\Omega \pair{\nabla\xi}{\nabla\xi}_x\d\mu_\eta(x)<\infty \biggr\}\;.
\end{align*}

As before one can give an equivalent definition that is useful for manifolds without boundary. The Sobolev space of sections of order $k\in\mathbb{R}^+$ is defined to be 
$$
\H^k(S):=\biggl\{ \xi\in\L^2(S) \Bigr| \norm{\xi}^2_k=\int_\Omega\pair{\xi(x)}{(1+\nabla^\dagger\nabla)^{k/2}\xi(x)}_x\d\mu_\eta(x)<\infty \biggr\}\;,
$$
where $\nabla^\dagger:\Gamma^\infty(T^*\Omega\otimes S)\to \Gamma^\infty(S)$ is the formal adjoint of the covariant derivative. Notice that this definition also holds for $k\notin\mathbb{N}$.\\

A natural domain where Dirac operators are closed, symmetric but not self-adjoint is 
$$
	\D_0=\H^1_0(S):=\overline{\Gamma_c^\infty(S)}^{\norm{\cdot}_1}\;.
$$
The subindex in $\H^1_0(S)$ is there to stress that the closure of the compactly supported sections of $S$ with respect to the Sobolev norm of order $1$ is not a dense subset of $\H^1(S)$, even though it is a dense subset of $\L^2(S)$\,.

The Dirac operator $D$ in the domain $\D_0$, denoted as the couple $(D,\D_0)$, has as adjoint operator the Dirac operator defined on the full Sobolev space of order $1$, $\H^1(S)$, namely $(D^\dagger,\D_0^\dagger)=(D,\H^1(S))$\,. Any self-adjoint extension of $(D,\D_0)$ must be a restriction of $(D,\H^1(S))$ to a domain $\D_{s.a.}$ that satisfies 
$$\D_0\subset\D_{s.a.}=\D_{s.a}^\dagger\subset \D_0^\dagger\;. $$

In order to determine such domains we are going to follow the successful approach set out in \cite{asorey05}. First we are going to need the Green's formula for the Dirac operator. Without loss of generality one can always pick an orthonormal frame $\{e_i\}\subset\Gamma^\infty(S)$ that it self-parallel, i.e. $$\pair{e_i}{e_j}_x=\delta_{ij}\;,\quad \nabla_{e_i}e_j=0\;.$$

Consider two fixed sections $\xi,\zeta\in\L^2(S)$ of the vector bundle $S$. For each point $p\in\Omega$ one can define a continuous linear functional on $T_p\Omega$ using Clifford multiplication as follows 
$$
\map[L_{p,\xi,\zeta}]{T_p\Omega}{\mathbb{R}}{V}{\pair{\xi}{V\cdot\zeta}_p}\;.
$$
Then, according to Riesz's Representation Theorem, it exists a unique $X\in T_p\Omega$ such that 
$$
	\eta_p(X,V)=-\pair{\xi}{V\cdot\zeta}_p\;,\forall V\in T_p\Omega\,,\; \xi,\zeta\in\L^2(S)\;.
$$

Now if we denote the Lie derivative by $$\L: T\Omega\times \Gamma^\infty(\Omega)\to\Gamma^\infty(\Omega)$$ we have
\begin{align*}
	\sum_{i}\mathcal{L}_{e_i}\pair{\xi}{e_i\cdot\zeta} = \sum_ie_i\pair{\xi}{e_i\cdot\zeta}  &=\sum_i\pair{\nabla_{e_i}\xi}{e_i\cdot\zeta}+\sum_i\pair{\xi}{e_i\cdot\nabla_{e_i}\zeta}\\
	&=-\sum_i\pair{e_i\cdot\nabla_{e_i}\xi}{\zeta}+\sum_i\pair{\xi}{e_i\cdot\nabla_{e_i}\zeta}\\
	&=-\pair{D\xi}{\zeta}+\pair{\xi}{D\zeta}\;,
\end{align*}
where we have used (d), (e), (f) of Section \ref{sec:Dirac}. On the other hand we have that
$$
	\sum_{i}\mathcal{L}_{e_i}\pair{\xi}{e_i\cdot\zeta}=-\sum_i \mathcal{L}_{e_i}\eta(X,e_i)=-\operatorname{div} X\;.
$$
Now let $i:\pO \to\Omega$ be the canonical inclusion mapping, let $\nu\in \mathfrak{X}(\Omega)$ be the normal vector field to the boundary and let $\theta\in\Lambda^1(\Omega)$ such that  
$$\d\mu_\eta=\theta\wedge\d\mu_{\partial\eta}\;,\quad \theta(X)=\eta(X,\nu)\;.$$ Then we have that
\begin{align*}
	\scalar{D\xi}{\zeta}-\scalar{\xi}{D\zeta} &= \int_\Omega \operatorname{div}X \d\mu_\eta(x)\\
		&=\int_\Omega \mathcal{L}_X(\d\mu_\eta)\\
		&=\int_\Omega \d(\mathrm{i}_X \d\mu_\eta)\\
		&=\int_{\pO} i^*(\mathrm{i}_X \d\mu_\eta)\\
		&=\int_{\pO} i^*(\mathrm{i}_X \theta\wedge\d\mu_{\partial\eta})\\
		&=\int_{\pO} i^*(\eta(X,\nu)\wedge\d\mu_{\partial\eta})\\
		&=\int_{\pO}\eta(X,\nu)\d\mu_{\partial\eta}=-\int_{\pO}\pair{\xi}{\nu\cdot\zeta}\d\mu_{\partial\eta}=\int_{\pO}\pair{\nu\cdot\xi}{\zeta}\d\mu_{\partial\eta}\;.
\end{align*}
In the right hand side of the last equation one has to understand that $\xi$ and $\zeta$ are the restrictions to the boundary of the corresponding sections. We use the same symbol whenever there is no risk of confusion. So far the Green's formula above is purely formal and holds for smooth sections. In order to extend it to the appropriate Hilbert spaces we need to make use again of Theorem \ref{thm:lions}. We reformulate it here for the case of vector bundles.

\begin{theorem}
	The restriction map 
$$
	\map[b]{\Gamma^\infty(S)}{\Gamma^\infty(\partial S)}{\xi}{\xi|_{\pO}}
$$ can be extended to a continuous, surjective map 
$$
	b\,:\, \H^{k}(S)\to\H^{k-1/2}(\partial S)\,\quad k>1/2\;.
$$
\end{theorem}

Applying the trace theorem to the Green's formula for the Dirac operator we have that it can be extended to 
\begin{equation}\label{eq:green}
	\scalar{D\xi}{\zeta}-\scalar{\xi}{D\zeta}=\int_{\pO}\pair{\nu\cdot\varphi}{\psi}\d\mu_{\partial\eta}=\scalar{J\cdot\varphi}{\psi}_{\pO}\;,
\end{equation}
where $\xi\in\H^1(S)$, $\zeta\in\H^1(S)$ $\varphi:=b(\xi)\in \H^{1/2}(S)$, $\psi:=b(\zeta)\in \H^{1/2}(S)$ and $J$ is the extension to $\H^{1/2}(S)$ of Clifford multiplication by the normal vector $\nu$\,.

The space of self-adjoint extensions of the Dirac operator $\mathcal{M}(D)$ is characterized by the set of maximally isotropic subspaces of the boundary form 
$$\Sigma(\varphi,\psi)=\scalar{J\cdot\varphi}{\psi}_{\pO}\;.$$
This boundary form satisfies that 
$$\overline{\Sigma(\varphi,\psi)}=-\Sigma(\psi,\varphi)\;.$$

Notice that $J\cdot J= -\1$ and therefore the Hilbert space at the boundary carries a natural polarization $\L^2(\partial S)=\H_+\oplus\H_-$\,, in terms of the proper subspaces of $J$,
$$\H_{\pm}=\{\varphi\in\L^2(\partial S) \mid J\varphi= \pm i \varphi \}\;,$$
and each section $\varphi\in\L^2(\partial S)$ has a unique decomposition $\varphi=\varphi_+ + \varphi_-$ with $\varphi_\pm\in\H_\pm$\,.

\begin{theorem}\label{thm:diracbc}
	Isotropic subspaces of the boundary form $\Sigma$ are in one-to-one correspondence with graphs of unitary maps $U:\H_+\to\H_-$\;.
\end{theorem}

\begin{proof}
Recall that the subspaces $\H_{\pm}$ are mutually orthogonal. Then one has that
$$
	\scalar{J\cdot \varphi}{\psi}_{\pO}=\scalar{J\cdot (\varphi_++\varphi_-)}{\psi_++\psi_-}_{\pO}	=-i\scalar{\varphi_+}{\psi_+}_{\pO}+i\scalar{\varphi_-}{\psi_-}_{\pO}\;.
$$
It is well known, cf. \cite{kochubei75,asorey05,ibortlledo13}, that maximally isotropic subspaces of a bilinear form like the one at the right hand side are determined uniquely by unitaries $U:\H_+\to\H_-$\,.
\end{proof}

Hence the space of self-adjoint extensions of the Dirac operator $\mathcal{M}(D)$ is in one-to-one correspondence with the set of unitary operators $\mathcal{U}(\H_+,\H_-)$ that preserve the Sobolev space of sections of order $1/2$, $\H^{1/2}(\partial S)$. We denote such set $\mathcal{U}(\H_+,\H_-)_{\H^{1/2}(\partial S)}$\;. Notice that any unitary operator defines a maximally isotropic subspace of the boundary form. However, if $U(\H_+\cap\H^{1/2}(\partial S))\not\subset\H^{1/2}(\partial S)$ then the trace map $b:\H^1(S)\to\H^{1/2}(\partial S)$ can be ill-defined.

\begin{example}
Let us consider that the Riemannian manifold is the unit disk $\Omega=\mathfrak{D}$. We can now endow the disk with the $\mathrm{Cl}(\mathfrak{D})-$bundle structure that we considered in Example \ref{ex:diracstructureexample2}. The outer unit normal vector field to the boundary, $ \nu\in T_{p}\mathfrak{D}\,,p\in S^1=\partial \mathfrak{D}$, can be expressed in terms of the orthonormal basis in $\mathbb{R}^2$ as
$$\nu=\cos\theta e_1 + \sin\theta e_2\;.$$
Its representation as an element of the Clifford algebra is
$$\nu=\cos\theta 
		\begin{bmatrix}
			0 & i \\ i & 0 
		\end{bmatrix} 
		+ \sin\theta 
		\begin{bmatrix}
			0 & -1 \\ 1 & 0
		\end{bmatrix}=
		\begin{bmatrix}
			0 & ie^{i\theta} \\ ie^{-i\theta} & 0
		\end{bmatrix}\;.
		$$
The proper subspaces of $J$ then become 
$$
\H_+=\operatorname{span}\left\{|+\rangle= \frac{1}{\sqrt{2}}
		\begin{pmatrix}
			e^{i\theta/2} \\ e^{-i\theta/2}
		\end{pmatrix}\right\}\;,
$$
$$
\H_-=\operatorname{span}\left\{|-\rangle= \frac{1}{\sqrt{2}}
		\begin{pmatrix}
			-e^{i\theta/2} \\ e^{-i\theta/2}
		\end{pmatrix}\right\}\;.
$$
Finally, the space of self-adjoint extensions $\mathcal{M}(D)$ is determined by the set of unitary operators $$U:\H^{1/2}(S^1)\otimes|+\rangle \to \H^{1/2}(S^1)\otimes|-\rangle\;.$$
\end{example}


\newpage

\section{Lecture IV.  The not semibounded case: New ideas and examples}\label{sec:notsemibounded}

Now that we have discussed how to define self-adjoint extensions of Dirac operators in terms of boundary conditions, we would like to have tools like those introduced in Lecture II that would allow us to avoid the analytic difficulties inherent to boundary value problems.  

Semiboundedness is a fundamental assumption in Kato's representation theorem and unfortunately the Dirac operators are not semibounded. Is it possible to find a generalization of Kato's representation theorem that can be applied to the generic case of non-semibounded operators? 

Given a generic, self-adjoint operator $T$ with domain $\D(T)$ one can always define a quadratic form associated to it using the spectral theorem. Indeed, consider the spectral decomposition of the self-adjoint operator $T =\int_{\mathbb{R}}\lambda \d E_\lambda$.  
The spectral theorem, cf. \cite{reed78}, ensures that any self-adjoint operator has a decomposition like the one above in terms of a resolution of the identity that in the case of self-adjoint operators with discrete spectrum reduces to the simple familiar expression $T = \sum_i\lambda_i P_i$.

The domain of the self-adjoint operator $T$ can be described using the spectral theorem as
\begin{equation}\label{eq:domainT}
	\D(T)=\left\{ \Phi\in\H \Bigr| \int_{\mathbb{R}}|\lambda|^2 \d(\scalar{\Phi}{E_\lambda\Phi})<\infty \right\}\;.
\end{equation}
However one can define the following domain
\begin{equation}\label{eq:domainQ}
	\D(Q_T)=\left\{ \Phi\in\H \Bigr| \int_{\mathbb{R}}|\lambda| \d(\scalar{\Phi}{E_\lambda\Phi})<\infty \right\}\;,
\end{equation}
that satisfies $\D(T)\subset\D(Q_T)$\,. Now for any $\Phi\in\D(Q_T)$ we can consider the quadratic form
\begin{equation}\label{eq:int_quadratic}
	Q_T(\Phi):=\int_\mathbb{R}\lambda \d(\scalar{\Phi}{E_\lambda\Phi})\;.
\end{equation}
This quadratic form is clearly represented by the operator $T$. Hence, to any self-adjoint operator one can always associate a quadratic form possessing an integral representation like \eqref{eq:int_quadratic}.  However, if the operator $T$ is not semibounded, one can not use Kato's representation theorem to recover the operator $T$ from its associated quadratic form. The rest of this section is devoted to obtain a generalization of Kato's representation theorem that can be applied to the not semibounded case.


\subsection{Partially orthogonally additive quadratic forms}

Suppose that the Hilbert space can be decomposed as the orthogonal sum of two subspaces 
$$\H=W_+\oplus W_-\;.$$
Let $P_\pm$ denote the orthogonal projections onto the subspaces $W_\pm$\,. 

\begin{definition}
We say that a quadratic form $Q$ with domain $\D$ is partially orthogonally additive with respect to the decomposition $\H=W_+\oplus W_-$\,, if $P_\pm\D\subset\D$ and whenever
$$
	\scalar{P_+\Phi}{P_+\Psi}=0\quad\text{and}\quad\scalar{P_-\Phi}{P_-\Psi}=0\
$$
we have that
\begin{equation}\label{eq:poa}
	Q(\Phi+\Psi)=Q(\Phi)+Q(\Psi)\,.
\end{equation}

We will call the quadratic forms $Q_\pm(\Phi):=Q(P_\pm\Phi)$ the sectors of the quadratic form. Notice that from the definition of partial orthogonal additivity we have that 
$$Q(\Phi)=Q(P_+\Phi+P_-\Phi)=Q(P_+\Phi)+Q(P_-\Phi)=Q_+(\Phi)+Q_-(\Phi)$$
\end{definition}

Recall now that for semibounded quadratic forms one can always define an associated graph norm. For a lower semibounded quadratic form there exists $a>0$ such that $Q(\Phi)\geq -a\norm{\Phi}^2$\,. Then we define 
$$
	\normm{\Phi}^2_Q=(1+a)\norm{\Phi}^2+Q(\Phi)\;.
$$
For an upper semibounded quadratic form there exists $b>0$ such that $Q(\Phi)<b\norm{\Phi}^2$ and we define
$$
	\normm{\Phi}^2_Q=(1+b)\norm{\Phi}^2-Q(\Phi)\;.
$$

\begin{theorem}\label{thm:notsemibounded}
	Let the quadratic form $Q$ with dense domain $\D$ be a partially orthogonally additive quadratic form with respect to the decomposition $\H=W_+\oplus W_-$\,. Let each sector $Q_\pm$ of $Q$ be either semibounded from below or from above and closable. Then there exists a unique self-adjoint operator $T$, with domain $\D(T)$ such that 
$$
	Q(\Phi,\Psi)=\scalar{\Phi}{T\Psi}\,,\quad\Phi\in\D,\Psi\in\D(T)\;.
$$
\end{theorem}

\begin{proof}
	If the sectors $Q_\pm$ are both simultaneously lower semibounded or upper semibounded then the quadratic form itself is semibounded and a direct application of Kato's theorem proves the statement. Hence consider that $Q_+$ is lower semibounded and $Q_-$ is upper semibounded. First notice that we can use the graph norms of the sectors to define a norm for the linear space $\D$, namely
$$
	\normm{\Phi}^2_Q=(1+a)\norm{P_+\Phi}^2+Q_+(\Phi)+(1+b)\norm{P_-\Phi}^2-Q_-(\Phi)\;.
$$
It is immediate to check that $\normm{\Phi}_Q<\infty$ for all $\Phi\in\D$. Therefore $\D$ can be closed with respect to this norm. Let us denote 
$$
	\overline{\D}=\overline{\D}^{\normm{\cdot}_Q}\;.
$$
Moreover, the quadratic form $Q$ can be continuously extended to this domain and we shall denote the extension by $\overline{Q}$\,. The extended quadratic form $\overline{Q}$ is also partially orthogonally additive with respect to $\{W_\pm\}$ and its sectors are closed, semibounded quadratic forms. Kato's representation theorem ensures that there exist two self-adjoint operators $T_\pm:\D(T_\pm)\to W_\pm$ such that 
$$
	Q_{\pm}(\Phi,\Psi)=\scalar{\Phi}{T_\pm\Psi}\,,\quad\Phi\in\overline{D}\,,\Psi\in\D(T_\pm)\;.
$$
Hence we can define the domain
\begin{equation}\label{eq:domainrepT}
	\D(T)=\{\Phi\in\H | P_+\Phi\in\D(T_+)\,,P_-\Phi\in\D(T_-)\}\;,
\end{equation}
and the operator 
$$
	T\Phi:=T_+P_+\Phi+T_-P_-\Phi\;.
$$
It is immediate to verify that T is self-adjoint.
\end{proof}

\begin{example}[\, \, Position operator] \label{ex: position}
Consider the Hilbert space $\H=\L^2(\mathbb{R})$\,. Let $Q$ be the Hermitean quadratic form with domain $\D=\C^{\infty}_c(\mathbb{R})$ defined by 
	$$Q(\Phi)=\int_{\mathbb{R}}\bar{\Phi}(x)\, x\, \Phi(x)\d x\;.$$
Consider the subspaces 
	$$W_{+}=\{\Phi\in\H|\operatorname{supp}\Phi\subset\mathbb{R}^+\}\;,$$
	$$W_{-}=\{\Phi\in\H|\operatorname{supp}\Phi\subset\mathbb{R}^-\}\;.$$ 
The projections onto these subspaces are given by multiplication with the corresponding characteristic functions. The sectors of the quadratic form above are therefore 
	$$Q_+(\Phi)=\int_{\mathbb{R}^+}\bar{\Phi}(x)x\Phi(x)\d x\;,$$
	$$Q_-(\Phi)=\int_{\mathbb{R}^-}\bar{\Phi}(x)x\Phi(x)\d x\;.$$
The graph norms of the sectors become in this case
$$
	\normm{P_+\Phi}^2_{Q_+}=\norm{P_+\Phi}^2+Q_+(\Phi)\,,
$$
$$
	\normm{P_-\Phi}^2_{Q_-}=\norm{P_-\Phi}^2-Q_-(\Phi)\,,
$$
and we get for the norm associated to $Q$
$$
	\normm{\Phi}^2_Q=\norm{\Phi}^2+\int_\mathbb{R}|x|\,|\Phi(x)|^2\d x\;.
$$
Hence, the closure of the domain $\D=\C^{\infty}_c(\mathbb{R})$ with respect to this norm becomes the domain of the closed quadratic form $Q$ that agrees with the domain defined by \eqref{eq:domainQ}. Now we want to obtain the domain of the self-adjoint operator representing $Q$. We need first to characterize the domains of the self-adjoint operators $T_\pm$\,. Again, we are going to use the characterization provided by Lemma \ref{lem:characterizationdavies}. According to this theorem $\Phi\in W_+$ is an element of the domain $\D(T_+)$ if it is an element of $P_+\D$ and it exists $\chi\in W_+$ such that
$$
	Q_+(\Psi,\Phi)=\scalar{\Psi}{\chi}\,,\quad \forall \Psi\in P_+\D\;.
$$
Rewriting this equality we get that 
$$
	\int_{\mathbb{R}^+}x\bar{\Psi}(x)\Phi(x) \d x = \int_{\mathbb{R}^+}\bar{\Psi}(x)\chi(x) \d x\,,\quad \forall \Psi\in P_+\D
$$
Since $P_+\D$ is dense in $W_+$ this implies that $x\Phi(x)=\chi(x)\in W_+$ and hence 
$$
	\int_{\mathbb{R}^+} |x|^2|\Phi(x)|^2 \d x <\infty \;,
$$
since $\chi$ must be in $\L^2(\mathbb{R}^+)\,.$
A similar argument for $T_-$ and having into account the definition in Eq.~\eqref{eq:domainrepT} shows that the domain of the representing operator $T$ is precisely
$$
	\D(T)=\left\{ \Phi\in\L^2(\mathbb{R}) \Bigr| \int_{\mathbb{R}}|x|^2 |\Phi(x)|^2 \d x<\infty \right\}\;,
$$
which agrees with the domain of the position operator.
\end{example}

Now we are going to consider the example of the momentum operator. Of course these case can be reduced to the previous example by using the Fourier transform, which is a unitary operator on the space of square integrable functions, and transforms the momentum operator $i\frac{\d}{\d x}$ into the multiplication operator. However, in order to set the ground for the most important example of this section, the case for the Dirac operator, we will define all the spaces of functions in the so called \emph{position space}.

\begin{example}[\,\, Momentum operator]
Consider the Hilbert space $\H = \L^2(\mathbb{R}^2)$. Let $Q$ be the Hermitean quadratic form with domain $\D=\H^1(\mathbb{R})$ defined by 
$$
	Q(\Phi,\Psi)=\int_{\mathbb{R}}\bar{\Phi}(x)i\frac{\d \Psi}{\d x}(x) \d x\;.
$$
Consider the subspaces 
$$
	W_{\pm}=\left\{\Phi\in\L^2 \Bigr| \Phi(x)=\int_{\mathbb{R}}\hat{\Phi}(k)e^{ikx}\d k\,, \hat{\Phi}(k)\in \L^2(\mathbb{R}^\pm)\right\}\;.
$$
We need to verify that these two subspaces are orthogonal and that the quadratic form above is partially orthogonally additive with respect to them. Let $\Phi_+\in W_+$ and $\Phi_-\in W_-$\;.
\begin{align*}
	\scalar{\Phi_+}{\Phi_-}&=\int_{\mathbb{R}^3}\bar{\hat{\Phi}}_+(k')\hat{\Phi}_-(k)e^{i(k-k')x}\d k \d k' \d x\\
		&=\int_{\mathbb{R}}\bar{\hat{\Phi}}_+(k)\hat{\Phi}_-(k)\d k =0\;.
\end{align*}
The last inequality follows because the supports of both functions are disjoint.

Similarly we have that 
$$
	Q(\Phi_+,\Psi_-)= -\int_{\mathbb{R}}k\bar{\hat{\Phi}}_+(k)\hat{\Phi}_-(k)\d k =0\;.
$$
Hence, the quadratic form above is partially orthogonally additive with respect to the decomposition above. Moreover, it is immediate to verify that the sectors are also semibounded. We can therefore obtain a unique, self-adjoint operator associated to it according to Theorem \ref{thm:notsemibounded}. This operator agrees with the momentum operator in Quantum Mechanics.
\end{example}

The two self-adjoint operators obtained in the two examples above are in fact the unique self-adjoint extensions on the real line of the symmetric operators defined by the position operator and the momentum operator. However, we have obtained them in a non standard way by closing the associated quadratic forms. The next section is devoted to show that in the case of the Dirac operator defined in a manifold with boundary we can use the boundary conditions defined in terms of a generic unitary operator to characterize self-adjoint extensions of it uniquely. The advantage is that one can ignore the restrictions related with the boundary data being elements of $\H^{1/2}(\partial\Omega)$, cf. Theorem \ref{thm:diracbc}.


\subsection{The Dirac Hamiltonian on $B^3$}

For the sake of simplicity we are going to consider as our base manifold the three dimensional unit ball, $\Omega=B^3$\,. As the Riemannian metric we are going to consider the restriction of the euclidean metric to the interior of ball. The results of this section can be extended straightforwardly to any bounded region in $\mathbb{R}^3$\,. As our $\mathrm{Cl}(B^3)-$bundle we are going to take the trivial bundle $S=\mathbb{C}^2\times B^3$\,,
$$
	\pi: S\to B^3\;,
$$
with flat Hermitean product given by the Hermitean product on $\mathbb{C}^2$. We need to choose a representation of the Clifford algebra $\mathrm{Cl}(\mathbb{R}^3)$ and we are going to take the representation $\gamma(e_j)=i\sigma_j$ given in terms of the Pauli matrices 
$$
\sigma_1=\begin{bmatrix} 0 & 1 \\ 1 & 0 \end{bmatrix}
\,,\quad
\sigma_2=\begin{bmatrix} 0 & i \\ -i & 0 \end{bmatrix}
\,,\quad
\sigma_3=\begin{bmatrix} 1 & 0 \\ 0 & -1 \end{bmatrix}
\;.
$$
Notice that, as they should, the generators satisfy 
$$
\gamma(e_i)\gamma(e_j)+\gamma(e_j)\gamma(e_i)=-2\mathbb{I}_{\mathbb{C}^2}\delta_{ij}\;.
$$
As the Hermitean connection we choose the flat connection defined by the componentwise derivatives. So finally we get that our Dirac operator is,
$$
D=\sum_{j=1}^3 i\sigma_j \frac{\partial}{\partial x^j}\;.
$$

As explained in Lecture III we can impose boundary conditions of the Dirac operator as follows. Let $\nu$ be the normal vector field to the boundary and consider the $\pm i$ eigenspaces $\H_\pm\subset \L^2(\partial S)$ of the operator defined by Clifford multiplication by $\nu$, that is:
$$
	\gamma(\nu)\Psi_\pm=\pm i\Psi_\pm\;; \qquad \Psi_\pm \in \mathcal{H}_\pm \, .
$$
Now, for any unitary operator $U:\H_+\to\H_-$, the boundary conditions $\varphi_-=U\varphi_+$ define domains where the Dirac operator is symmetric. In order to guarantee self-adjointness one needs to impose extra conditions on the unitary operator $U$. However, we can use the approach introduced in this section and the former to characterize uniquely a self-adjoint extension that satisfies the boundary condition. We only need to find a couple of subspaces that guarantee that the associated quadratic form is partially orthogonally additive with respect to them.\\

Consider the partial differential equation 
$$
D\Psi=E\Psi\,,\quad\Psi\in\H^1(S)\;.
$$

We are going to try solutions of the form $\Psi(\bx)=\xi e^{i\bk\cdot\bx}$\,, $\xi\in\mathbb{C}^2$\;. A short calculation shows that 
$$
	D\Psi=-\sum_{j=1}^3k_j\sigma_j\xi e^{i\bk\cdot\bx}=E\xi e^{i\bk\cdot\bx}
$$
Hence $\Psi$ is  a solution if $\xi$ is an eigenvalue of the Hermitean matrix $-\sum_{j}k_j\sigma_j$\,. From the properties of the Clifford algebra we have that $\left(-\sum_{j}k_j\sigma_j\right)^2=\sum_j|k_j|^2=:\bk^2$. Hence the matrix has two eigenspaces with values $\pm \sqrt{\bk^2}$\,. 
We can now select functions
$$
	\map[\xi_\pm]{\mathbb{R}^3}{\mathbb{C}^2}{\bk}{\xi_\pm(\bk)}
$$
such that for each fixed $\bk$ we have that
$$
-\sum_{j}k_j\sigma_j\xi_\pm(\bk)=\pm \sqrt{\bk^2} \xi_\pm(\bk)\;.
$$
Notice that such functions are orthogonal to each other for each value $\bk$\, since they are eigenvectors corresponding to different eigenvalues. Now we define the following subspaces of $\L^2(\mathbb{R}^3)$.
$$
	\mathfrak{h}=\{ \phi\in\L^2(\R^3) \,|\, \operatorname{supp}\phi\subset B^3 \}\;,
$$
$$
	\hat{\mathfrak{h}}=\left\{\hat{\phi}\in\L^2(\R^3) \,|\, \hat{\phi}(\bk)=\frac{1}{(2\pi)^{3/2}}\int_{\R^3}\phi(\bx)e^{-i\bk\cdot\bx} \d\bx \,;\phi\in\mathfrak{h} \right\}\;.
$$
These subspaces are closed subspaces of $\L^2(\R^3)$\,. Now we are ready to define a decomposition of $\L^2(S)$\,. Consider the closed subspaces of $\L^2(S)$ defined by
$$
	W_{\pm}=\left\{ \Phi\in\L^2(S) \Bigr| \Phi(\bx)=\frac{1}{(2\pi)^{3/2}}\int_{\R^3}\xi_{\pm}(\bk)\hat{\phi}(\bk)e^{i\bk\cdot\bx}\d \bk\,; \hat{\phi}\in\hat{\mathfrak{h}} \right\}\;.
$$
As before we have to show that they are orthogonal and that the quadratic form $Q=\scalar{\Phi}{D\Psi}$ is partially orthogonally additive with respect to them. Have in mind that $Q$ is a Hermitean quadratic form only if $D$ is defined on a symmetric domain.

Let $\Phi_+\in W_+$ and $\Phi_-\in W_-$.
\begin{align*}
	\scalar{\Phi_+}{\Phi_-}&=\int_{B^3}\int_{\R^6}\pair{\xi_+(\bk')}{\xi_-(\bk)}\bar{\hat{\phi}}_+(\bk')\hat{\phi}_-(\bk)e^{i(\bk-\bk')\cdot\bx} \d \bk' \d \bk \d \bx\\
		&=\int_{\R^3} \pair{\xi_+(\bk)}{\xi_-(\bk)}\bar{\hat{\phi}}_+(\bk)\hat{\phi}_-(\bk)\d \bk=0\;.
\end{align*}
The last equality follows from the fact that $\xi_+$ and $\xi_-$ are orthogonal for each $\bk$ and therefore $\pair{\xi_+(\bk)}{\xi_-(\bk)}\equiv0$\;.

Similarly we have that
\begin{align*}
	\scalar{\Phi}{D\Psi}&=\sum_j\int_{B^3}\int_{\R^6}\pair{\xi_+(\bk')}{i\sigma_j\xi_-(\bk)}\bar{\hat{\phi}}_+(\bk')\hat{\phi}_-(\bk)\frac{\partial}{\partial x^j}e^{i(\bk-\bk')\cdot\bx} \d \bk' \d \bk \d \bx\\
		&=\int_{B^3}\int_{\R^6}\pair{\xi_+(\bk')}{-\sum_jk_j\sigma_j\xi_-(\bk)}\bar{\hat{\phi}}_+(\bk')\hat{\phi}_-(\bk)e^{i(\bk-\bk')\cdot\bx} \d \bk' \d \bk \d \bx \\
		&=\int_{B^3}\int_{\R^6}\pair{\xi_+(\bk')}{-\sqrt{\bk^2}\xi_-(\bk)}\bar{\hat{\phi}}_+(\bk')\hat{\phi}_-(\bk)e^{i(\bk-\bk')\cdot\bx} \d \bk' \d \bk \d \bx=0
\end{align*}
The last inequality follows again from the fact that $\xi_\pm$ are orthogonal. Using very similar calculations it is easy to show that one sector is positive and that the other is negative and hence we can recover for each boundary condition $\varphi_-=U\varphi_+$ a unique self-adjoint extension.


\newpage

\section{Lecture V: Symmetries and self-adjoint extensions}

In this lecture we study what role do quantum symmetries play when one faces the problem of selecting self-adjoint extensions.  It is not necessary to emphasize here the fundamental role that symmetries play in Quantum Mechanics, however it is relevant to point out that the presence of symmetries can restrict the possible self-adjoint extensions of a given symmetric operator, that is, contrary to what one could guess, not all self-adjoint extensions are compatible with the quantum symmetries that the symmetric problem has.

In order to fix the ideas we first show an example where these two notions arise together.

\begin{example}\label{ex:symmetry} Consider as Riemannian manifold $\Omega$ the unit ball in $\mathbb{R}^3$.  The boundary of this manifold is the unit sphere $S^2$.  Now consider that as symmetry group $G$ we have the unitary group $\mathcal{U}(1)$ acting by rotations around the $z-$axis.  One can consider now the quotient space $\Omega/G$ and its boundary $\partial\negthickspace\left( \Omega/G\right)$ (see Figure \ref{fig:symmetryball}).   Notice that the action of $G$ restricts to an action on $\partial \Omega$.

The boundary of $\Omega/G$ has two pieces.  One is the quotient of the original boundary $\partial \Omega = S^2$ by the symmetry group, i.e. $S^2/G$.  We shall call it the regular part of the quotient's boundary and denote it as $\partial_{\textrm{reg}}(\Omega/G)$. The other part, that corresponds to the $z$-axis, will be called the singular part of the boundary and will be denoted as $\partial_{\textrm{sing}}(\Omega/G)$.  The singular part arises when the action of the group has fixed points.
\end{example}

Thus two natural questions arise:  When are going to be the self-adjoint extensions of the Laplace-Beltrami operator compatible with the action of $G$?   What are the self-adjoint extensions of the Laplace-Beltrami operator that can be defined on $\Omega/G$?

\begin{figure}[h!]
\center
\includegraphics[width=12cm]{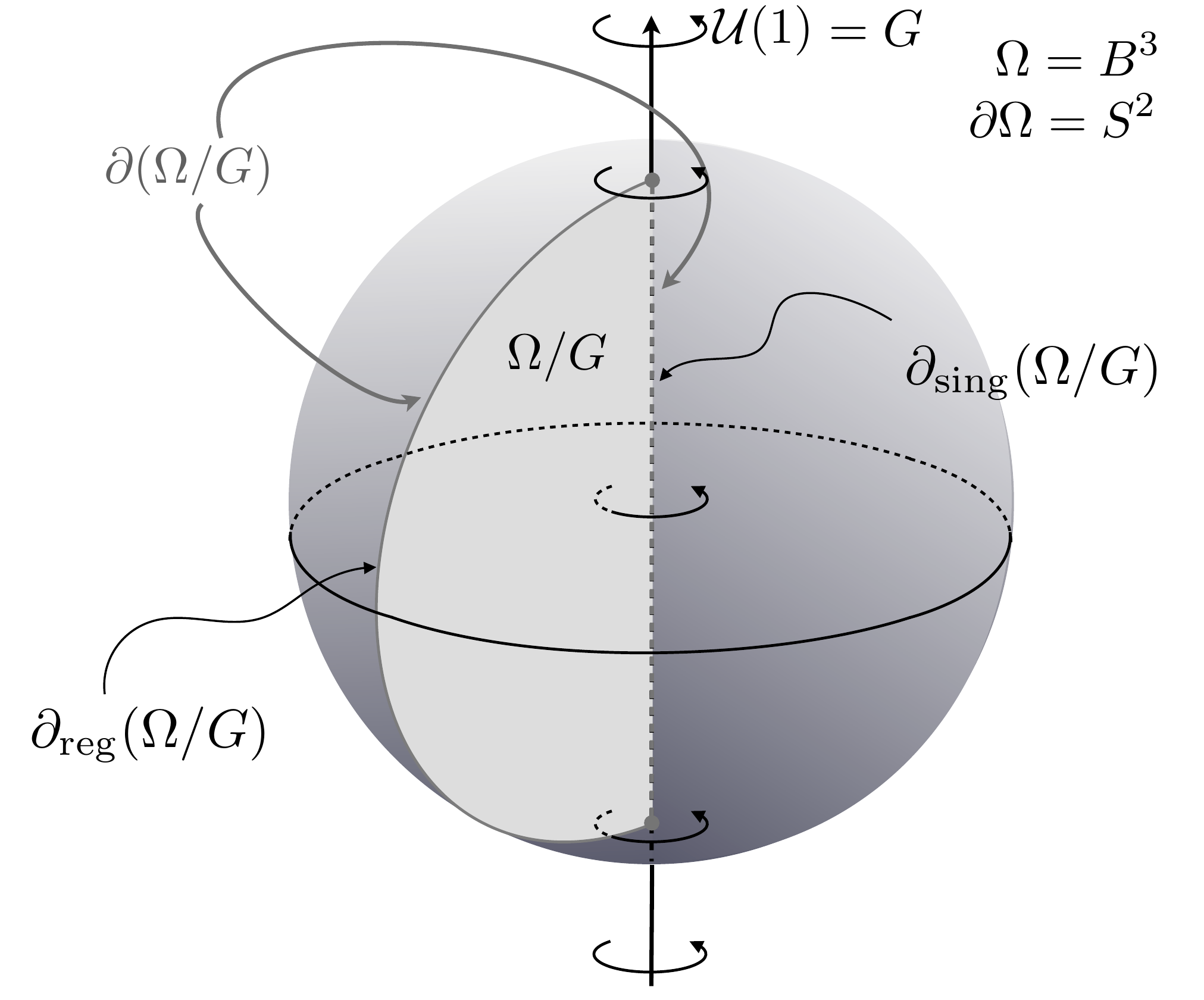}
\caption{\footnotesize{Solid ball  $B^3\subset\R^3$ under the action of the group $\mathcal{U}(1)$ acting by rotation on the $z$-axis.}}\label{fig:symmetryball}
\end{figure}

To avoid confusion between the terms symmetric operators and quantum symmetries we shall use the following terminology.  We will say that a symmetric operator is $G$-invariant if this operator possesses $G$ as a symmetry group. In what follows we will motivate the definition of $G$-invariant symmetric operators and show what role plays the $G$-invariance in the characterization of self-adjoint extensions. For further details we refer to \cite{ibortlledo14}.


\subsection{Self-adjoint extensions of $G$-invariant operators}

Suppose that we are given a group $G$ and a strongly continuous unitary representation of it, i.e. a map 
$$
	\map[V]{G}{\mathcal{U}(\H)}{g}{V(g)}\;,
$$
that verifies $V(g_1g_2)=V(g_1)V(g_2)$, $V(e)=\1_{\H}$ for all $g_1,g_2\in G$\,.

Now let $T$ be a self-adjoint operator with domain $\D(T)$. According to Stone's Theorem it defines a one-parameter unitary group 
$$U_t=e^{itT}\;.$$
That the operator is $G$-invariant means that the action of the group and the unitary evolution generated by $T$ commute 
$$
	V(g)U_t=U_tV(g) , , \qquad \forall g \in G, \, \,  t \in \mathbb{R} \;.
$$
The latter is equivalent to $V(g)\D(T)\subset\D(T)$ for all $g\in G$ and $V(g)T \Phi = TV(g) \Phi $ for all $g\in G$ and for all $ \Phi  \in\D(T)\,.$  This motivates the definition of $G$-invariance for any operator.

\begin{definition}
	Let $T$ be a densely defined operator with domain $\D(T)$\;. We will say that $T$ is $G$-invariant if 
	$$V(g)\D(T)\subset\D(T)\,,\quad \forall g\in G$$ and $$V(g)T\Phi=TV(g)\Phi\,\quad\forall g\in G\,, \forall \Phi\in \D(T)\;.$$
\end{definition}

We will need the following result, cf. Theorem \ref{thm:vonneumann}.

\begin{lemma}
 If $T$ is a $G$-invariant, densely defined, symmetric operator, then
	\begin{enumerate}
		\item[i.] $T^\dagger$ is $G$-invariant.
		\item[ii.] The deficiency spaces $\mathcal{N}_+$\,, $\mathcal{N}_-$ are $G$-invariant.
	\end{enumerate}
\end{lemma}

\begin{proof}
	The first implication follows by applying directly the definition of the adjoint operator and its domain. For (ii) it is enough to consider that $\Psi\in\mathcal{N}_+$ iff $(T^\dagger-i\1)\Psi=0$ and therefore using (i) we get
	$$
		0=V(g)(T^\dagger-i\1)\Psi=(T^\dagger-i\1)V(g)\Psi\;.
	$$
\end{proof}

Then the answers to our previous questions are contained in the following theorem:

\begin{theorem}
	Let $T_K$ be a von Neumann self-adjoint extension of the $G$-invariant symmetric operator $T$ determined by the unitary operator $K \colon \mathcal{N}_+\to\mathcal{N}_-$. Then $T_K$ is $G$-invariant if and only if, 
	$$
	[V(g),K]\xi=0 \, , \qquad \forall \xi\in\mathcal{N}_+\;.
	$$
\end{theorem}

\begin{proof}
	We will just show that the commutation $[V(g),K]\xi=0$ implies that $[T_K,V(g)]=0$. The full details of the proof can be found at \cite{ibortlledo14}. Recall von Neumann's characterization of self-adjoint extensions, Theorem \ref{thm:vonneumann}. We will use the fact that $V(g)\D(T_K)\subset\D(T_K)$ which is a direct consequence of $[V(g),K]\xi=0$ and the previous Lemma.
	 Then,
	 \begin{align*}
		V(g)T_K(\Phi_0+(\1+K)\xi_+)&=V(g)T\Phi_0 + V(g)(\1+K)\xi_+\\
			&=TV(g)\Phi_0 + V(g)(\1+K)\xi_+\\
			&=T_K(V(g)(\Phi_0+(\1+K)\xi_+) )\;.
	 \end{align*}
\end{proof}


\subsection{$G$-invariance and quadratic forms}

Once we have established the relationship between $G$-invariant symmetric operators and its $G$-invariant self-adjoint extensions it is natural to characterize the $G$-invariance properties at the level of quadratic forms. Moreover we want to know if we can use the same tools introduced in Lectures II and IV to characterize $G$-invariant self-adjoint extensions of the corresponding symmetric operators. Since quadratic forms take values in $\mathbb{C}$ a natural definition of $G$-invariance is the following.

\begin{definition}
	A quadratic form $Q:\D\subset\H\to\mathbb{C}$ is $G$-invariant if for all $g\in G$ and $\Phi\in\D$ we have that $V(g)\D\subset \D$ and $$Q(V(g)\Phi,V(g)\Psi)=Q(\Phi,\Psi) \, , \qquad \forall \Phi, \Psi \in \mathcal{D} \;.$$
\end{definition}

Now it is natural to ask what is the relationship that exists between the two notions of $G$-invariance that we have introduced so far. The answer to this question comes again from Kato's Representation Theorem (Theorem \ref{thm:kato}).

\begin{theorem}\label{thm:QAinvariant}
Let $Q$ be a closed, semi-bounded quadratic form with domain $\D$ and let $T$ be the representing semi-bounded, self-adjoint operator.
The quadratic form $Q$ is $G$-invariant if and only if the operator $T$ is $G$-invariant.
\end{theorem}

\begin{proof}
To prove the direct implication recall from
Lemma~\ref{lem:characterizationdavies} that $\Psi\in\D(T)$ if and only if $\Psi\in\D$ and there exists $\chi\in\H$ such that
$$Q(\Phi,\Psi)=\scalar{\Phi}{\chi}\quad\forall \Phi\in\D\;.$$
Then, if $\Psi\in\D(T)$, and using the $G$-invariance of the quadratic form, we have that
\begin{align*}
Q(\Phi,V(g)\Psi)&=Q(V(g)^\dagger\Phi,\Psi)\\
                &=\scalar{V(g)^\dagger\Phi}{\chi}=\scalar{\Phi}{V(g)\chi}\;.
\end{align*}
This implies that $V(g)\Psi\in\D(T)$ and from
$$TV(g)\Psi=V(g)\chi=V(g)T\Psi\,,\quad \Psi\in\D(T)\,,g\in G\;,$$
we show the $G$-invariance of the self-adjoint operator $T$.

To prove the converse we use the fact that $\D(T)$ is dense in $\D$ with respect to the graph norm $\normm{\cdot}_{Q}$\,.
For $\Phi,\Psi\in\D(T)$ we have that
\begin{align*}
 Q(\Phi,\Psi)&= \scalar{\Phi}{T\Psi}=\scalar{V(g)\Phi}{V(g)T\Psi} \\
             &= \scalar{V(g)\Phi}{TV(g)\Psi}=Q(V(g)\Phi,V(g)\Psi)\;.
\end{align*}

These equalities show that the $G$-invariance of $Q$ is true at least for the elements in the domain of the operator.
Now for any $\Psi\in\D$ there is a sequence
$\{\Psi_n\}_n\subset\D(T)$ such that $\normm{\Psi_n-\Psi}_Q\to 0$\,. This, together with the equality above,
implies that $\{V(g)\Psi_n\}_n$ is a Cauchy sequence with respect to $\normm{\cdot}_{Q}$\,. Since $Q$ is closed,
the limit of this sequence is in $\D$\,. Moreover it is clear that $\lim_{n\to\infty}V(g)\Psi_n=V(g)\Psi$\,,
hence
$$\normm{V(g)\Psi_n-V(g)\Psi}_Q\to0\,.$$

So far we have proved that $V(g)\D\subset\D$. Now for any $\Phi,\Psi\in\D$ consider
sequences $\{\Phi_n\}_n,\{\Psi_n\}_n\subset\D(T)$ that converge respectively to $\Phi,\Psi\in\D$ in the topology induced by $\normm{\cdot}_{Q}$\,. Then the limit
\begin{align*}
Q(\Phi,\Psi)&=\lim_{n\to\infty}\lim_{m\to\infty}Q(\Phi_n,\Psi_m)\\
            &=\lim_{n\to\infty}\lim_{m\to\infty}Q(V(g)\Phi_n,V(g)\Psi_m)=Q(V(g)\Phi,V(g)\Psi)\;.
\end{align*}
concludes the proof.
\end{proof}


\subsection{Symmetries and the Laplace-Beltrami operator}

In this section we are going to construct $G$-invariant self-adjoint extensions of the Laplace-Beltrami operator. In order to do so we will take the approach introduced in Section \ref{sec:LBrevisited}. That is, we will construct the quadratic form associated to the Laplace-Beltrami operator, Eq.~\ref{eq:greenLB} and we will select a domain for it such that it is semibounded, this is granted by the gap condition on the unitary operator $U$, cf. Definition~\ref{def:gap}, then we will impose further conditions on $U$ that ensure that the quadratic form is $G$-invariant. According to Theorem~\ref{thm:QAinvariant} the Friedrichs' extensions of the Laplace-Beltrami operator obtained this way will be $G$-invariant.

Let $U:\L^2(\pO)\to\L^2(\pO)$ be a unitary operator with gap, cf. Definition~\ref{def:gap}, and let $A_{U_{W^\bot}}$ be the partial Cayley transform associated to it. Recall that $A_{U_{W^\bot}}$ is a bounded self-adjoint operator. We will consider the quadratic form defined by 
$$Q_U(\Phi,\Psi)=\scalar{\d\Phi}{\d\Psi}-\scalar{\varphi}{A_{U_{W^\bot}}\psi}\;,$$
with domain $$\D=\{\Phi\in\H^1 | \varphi_W=P_W\varphi=0\}\;.$$
Notice that the boundary condition $\varphi-i\dot{\varphi}=U(\varphi+i\dot{\varphi})$ is implemented in this case in a separated way by means of the Eqs.~\ref{eq:bcqf}. Recall that $W$ is the closed subspace associated to the eigenvalue $\{-1\}$ of the unitary operator $U$ and that $P_W$ is the orthogonal projection onto it. 

The gap condition ensures that this quadratic form is semibounded from below. We are considering a slightly more general situation than the one considered in Section~\ref{sec:LBrevisited} and in order to ensure the closability of the quadratic form we will need to require that the unitary operator is well behaved with respect to the boundary data. As detailed in \cite{ibortlledo13} it is enough to consider that
$$A_{U_{W^\bot}}|_{\H^{1/2}(\pO)}:\H^{1/2}(\pO)\to\H^{1/2}(\pO)$$ is continuous.

Now we need to address the problem of characterizing those unitaries that lead to $G$-invariant quadratic forms. Of course, we need first to specify what kind of groups and group representations are we going to consider. Since the domain of the quadratic form is a dense subset of $\H^1(\Omega)$ we will need at least that the unitary representation $$V:G\to\L^2(\Omega)$$ of the group be such that it preserves $\H^1(\Omega)$, i.e., we will need that $$V(g)\H^1(\Omega)\subset\H^1(\Omega)\;.$$

\begin{definition}
	We will say that the unitary representation $V$ is traceable if there exists another unitary representation $V_{\pO}:G\to\L^2(\pO)$ such that 
	$$V_{\pO} \left(\H^{1/2}(\pO)  \right) \subset\H^{1/2}(\pO)$$ and makes the following diagram commutative
	$$
		\begin{array}{ccc}
		\H^{1}(\Omega)&\stackrel{V(g)}{\longrightarrow} &\H^1({\Omega})\\
		b \downarrow&&\downarrow b\\
		\H^{1/2}(\pO)&\stackrel{V_{\pO}(g)}{\longrightarrow} &\H^{1/2}({\pO})
		\end{array}
	$$	
\end{definition}

\begin{lemma}\label{lem:traceable}
	Let $V:G\to\L^2(\Omega)$ be a unitary representation of the group $G$ such that $V(g)\H^1(\Omega)\subset\H^1(\Omega)$\,. If in addition $V(g)\H^1_0(\Omega)\subset\H^1_0(\Omega)$ then $V$ is traceable.
\end{lemma}

\begin{proof}
	By Theorem~\ref{thm:lions} we have that the map $b \colon \H^1(\Omega)\to\H^{1/2}(\pO)$ is continuous and onto. Hence, by the inverse mapping theorem $$\H^1(\Omega)/\operatorname{ker}b\simeq \H^{1/2}(\pO)\;.$$ Noticing that $\H^1_0(\Omega)=\operatorname{ker}b$ ends the proof.
\end{proof}

We are ready to state the main result of this lecture. We are going to assume that the unitary operator does not have $\{-1\}$ in the spectrum. This assumption is a bit stronger than the gap assumption but will simplify the statement of the next theorem. For a generalization for unitary operators with gap and the proof we refer to \cite{ibortlledo14}.

\begin{theorem}\label{thm:invariantrep}
	Let $U$ be a unitary operator such that $\{-1\}$ is not in its spectrum and such that 
	$$
	U|_{\H^{1/2}(\pO)}:\H^{1/2}(\pO)\to\H^{1/2}(\pO)
	$$ 
	is continuous. Let $V:G\to\L^2(\Omega)$ be a traceable strongly continuous unitary representation preserving $\H^1(\Omega)$. The quadratic form $Q_U$ is $G$-invariant if and only if $$[U,V(g)]=0\;,\forall g\in G\;.$$
\end{theorem}

In order to show how to use the previous theorem to obtain $G$-invariant self-adjoint extensions, we will go back to the Example~\ref{ex:symmetry}.

\begin{example}  Let us consider again the situation in Example \ref{ex:symmetry}.  There the manifold $B^3$, see Figure~\ref{fig:symmetryball}, is invariant under rotations on the $z$-axis. Moreover, the group $\mathcal{U}(1)$ acts by isometric diffeomorphisms on the manifold. It is easy to show that any such action induces a unitary representation of the group that preserves all the Sobolev spaces, in particular $\H^1(\Omega)$. Moreover, the representation induced this way is always traceable since $\H^1_0(\Omega)$ is also preserved, cf. Lemma~\ref{lem:traceable}. 
	
Now consider that $U\in\mathcal{U}(\L^2(S^2))$ is such that 
	$$
	\map[U]{\H^{1/2}(S^2)}{\H^{1/2}(S^2)}{f(\theta,\varphi)}{(Uf)(\theta,\varphi)}\;.
	$$
We want to characterize those unitaries that commute with the action of the group $(V(\alpha)f)(\theta,\varphi)=f(\theta,\varphi+\alpha)$\, (cf. Theorem~\ref{thm:invariantrep}).

To analyse what are the possible unitary operators that lead to $G$-invariant quadratic forms it is convenient to use the Fourier 
series expansions of the elements in $\L^2(S^2)$\,. Let $f\in\L^2(S^2)$\,, then
$$f(\theta,\varphi)=\sum_{n\in\mathbb{Z}}\hat{f}_n(\theta)e^{\mathrm{i}n\varphi}\;,$$
where the coefficients of the expansion are given by $$\hat{f}_n(\theta)=\frac{1}{2\pi}\int_0^{2\pi}f(\theta,\varphi)e^{-\mathrm{i}n\varphi}\d\varphi\;.$$
We can therefore consider the induced action of the group $G$ as a unitary operator on $\L^2([0,\pi])\otimes{\ell}_2$\,. 
In fact we have that:
\begin{align*}
\widehat{v (\alpha)f}_n(\theta)&
=\frac{1}{2\pi}\int_0^{2\pi}f(\theta,\varphi+\alpha)e^{-\mathrm{i}n\varphi}\d\varphi\\
&=\sum_{m\in\mathbb{Z}}\hat{f}_m(\theta)e^{\mathrm{i}m\alpha}\int_0^{2\pi}\frac{e^{\mathrm{i}(m-n)\varphi}}{2\pi}\d\varphi
 =e^{\mathrm{i}n\alpha}\hat{f}_n(\theta)\;.
\end{align*}
This shows that the induced action of the group $G$ is represented by a unitary operator in $\mathcal{U}(\ell_2)$ that acts 
diagonally in the Fourier series expansion.
More concretely, we can represent it as $\widehat{v (\alpha)}_{nm}=e^{\mathrm{i}n\alpha}\delta_{nm}\,$\,.
From all the possible unitary operators acting on the Hilbert space of the boundary, only those whose representation in $\ell_2$
commutes with the above operator will lead to $G$-invariant quadratic forms, cf. Theorem~\ref{thm:invariantrep}.
Since $\widehat{v (\alpha)}$ acts
diagonally on $\ell_2$ it is clear that only operators of the form $\hat{U}_{nm}=u\otimes e^{\mathrm{i}\beta_n}\delta_{nm}$\,, 
$\{\beta_n\}_n\subset\mathbb{R}$\,, $u\in\mathcal{U}(\mathcal{L}^2[0,\pi])$ will lead to $G$-invariant quadratic forms.

As a particular case we can consider that all parameters are equal, i.e., $\beta_n=\beta$, $n\in\mathbb{Z}$\,.
In this case it is clear that $(\widehat{Uf})_n(\theta)=e^{\mathrm{i}\beta}(uf_n)(\theta)$\,, which gives the following 
admissible unitary with spectral gap at $-1$:
$$Uf=e^{\mathrm{i}\beta}uf\;.$$
Consider that $ug(\theta)=e^{\mathrm{i}\gamma(\theta)}g(\theta)$. This means that the unitary operator $u\in\mathcal{U}(\mathcal{L}^2[0,\pi])$ is of Robin type or, in other words, that it leads to Robin boundary conditions with respect to the variable $\theta$. The equation above shows that the only Robin functions compatible with the $\mathcal{U}(1)$-invariance are those that are defined by functions that do not depend on the angle $\varphi$\,.
\end{example}


\end{document}